%% file: necklace_paper1.tex
\documentclass[12pt]{article}
\usepackage{putex}
\usepackage{graphicx}
\usepackage{epstopdf}
\usepackage{enumerate}
\usepackage{cite}
\usepackage{amsthm}
\usepackage{verbatim}

\numberwithin{equation}{section}

\newcommand{\abs}[1]{\left\lvert #1 \right\rvert}

\newcommand {\be} {\begin {equation}}
\newcommand {\ee} {\end {equation}}

\newcommand {\bes} {\begin {equation*}}
\newcommand {\ees} {\end {equation*}}

\newcommand{\es}[2] {\begin{equation} \label{#1} \begin{split} #2 \end{split} \end{equation}}

\newcommand{\Z}{\mathbb{Z}}
\newcommand{\R}{\mathbb{R}}
\newcommand{\HH}{\mathbb{H}}
\newcommand{\C}{\mathbb{C}}

\def\tr{\operatorname{tr}}
\def\Vol{\operatorname{Vol}}
\def\Area{\operatorname{Area}}
\def\ker{\operatorname{Ker}}
\def\im{\operatorname{Im}}
\def\coker{\operatorname{Coker}}

\def\s{\sigma}
\def\g{\gamma}

\theoremstyle{plain}
\newtheorem*{YF}{Yee's Formula}
\newtheorem{cor}{Corollary}
\newtheorem*{TF}{Tree Formula}

\begin{document}

\preprint{PUPT-2374}

\institution{PU}{Joseph Henry Laboratories, Princeton University, Princeton, NJ 08544}

\title{From Necklace Quivers to the $F$-theorem, Operator Counting, and $T(U(N))$}

\authors{Daniel R.~Gulotta, Christopher P.~Herzog,  and Silviu S.~Pufu}

\abstract{
The matrix model of Kapustin, Willett, and Yaakov is a powerful tool for exploring the properties of strongly interacting superconformal Chern-Simons theories in 2+1 dimensions.
In this paper, we use this matrix model to study necklace quiver gauge theories with
${\mathcal N}=3$ supersymmetry and $U(N)^d$ gauge groups in the limit of large $N$.
 In its simplest application, the matrix model computes the free energy of the gauge theory on $S^3$.  The conjectured $F$-theorem states that this quantity should decrease under renormalization group flow.  We show that for a simple class of such flows, the $F$-theorem holds for our necklace theories.
We also provide a relationship between matrix model eigenvalue distributions and numbers of chiral operators that we conjecture holds more generally.
Through the AdS/CFT correspondence, there is therefore a natural dual geometric interpretation of the matrix model saddle point in terms of volumes of 7-d tri-Sasaki Einstein spaces and some of their 5-d submanifolds.
As a final bonus, our analysis gives us the partition function of the $T(U(N))$ theory on $S^3$.

}

\date{May 2011}

\maketitle

\tableofcontents

\newpage

\section{Introduction}

Exact results in strongly-interacting field theories are generally rare.  In supersymmetric field theories, supersymmetry places strong constraints on various properties of chiral operators, and exact results pertaining to these operators might be possible even at strong coupling.  For three-dimensional superconformal theories, recent progress in finding such exact results that hold at any coupling was made in \cite{Kapustin:2009kz, Jafferis:2010un,Hama:2010av}, where the partition function of superconformal theories on $S^3$ with ${\cal N} \geq 2$ supersymmetry, as well as the expectation values of certain BPS Wilson loops, were reduced from path integrals to finite-dimensional multi-matrix integrals.  This major simplification was achieved through the localization technique developed in \cite{Pestun:2007rz} for four-dimensional theories.  

A consequence of this work is the realization that the ``free energy'' $F$ defined as minus the logarithm of the path integral on $S^3$,
 \es{FDef}{
  F = -\log \abs{Z_{S^3}} \,,
 }
with appropriate subtractions of power law divergences, might represent a good measure of the number of degrees of freedom in any field theory, supersymmetric or non-supersymmetric.  One way in which $F$ can be thought of as a measure of the effective number of degrees of freedom is the conjecture made in \cite{Jafferis:2011zi} that $F$ decreases along renormalization group (RG) flows and is stationary at RG fixed points.  This conjecture was called the ``$F$-theorem'' in \cite{Jafferis:2011zi} and was tested in a few RG flows in large $N$ supersymmetric $U(N)$ gauge theories \cite{Jafferis:2011zi,Amariti:2011da}.  If the $F$-theorem is true, then $F$ would be a 3-d analog of the central charge $c$ from two-dimensional field theory, which is known to have the same monotonicity property along RG flows \cite{Zamolodchikov:1986gt}, and it would resemble the Weyl anomaly coefficient $a$ from 4-d theories, which is also believed to decrease along RG trajectories \cite{Cardy:1988cwa}.  Actually, the free energy $F$ also resembles $a$ in another way:  just like $a$, $F$ can be used to find the exact R-symmetry in the infrared (IR) by computing $F$ as a function of a set of trial $R$-charges and then maximizing it \cite{Jafferis:2010un}.  The analogous procedure in 4-d theories is called ``$a$-maximization'' \cite{Intriligator:2003jj}.

In the context of the AdS/CFT correspondence \cite{Gubser:1998bc,Maldacena:1997re,Witten:1998qj}, the free energy $F$ can also be computed holographically from the gravity side of the correspondence.  In particular, for a CFT dual to $AdS_4$ of radius $L$ and effective four-dimensional Newton constant $G_N$, $F$ is given by \cite{Emparan:1999pm}
 \es{FHol}{
  F = \frac{\pi L^2}{2 G_N} \,.
 }
It was shown in \cite{Casini:2011kv} that in any CFT $F$ can also be interpreted as an entanglement entropy between a disk and its complement in the $\R^{2, 1}$ theory.   In turn, this entanglement entropy equals the holographic $a_*$ function defined in \cite{Myers:2010tj,Myers:2010xs} that was shown to always decrease along holographic RG flows.

If the $AdS_4$ background mentioned above arises as a Freund-Rubin compactification $AdS_4 \times Y$ of M-theory, where $Y$ is a seven-dimensional Sasaki-Einstein space threaded by $N$ units of four-form flux, then the quantization of the $AdS_4$ radius in Planck units implies that at large $N$ eq.~\eqref{FHol} becomes \cite{Herzog:2010hf}
  \es{MtheoryExpectation}{
  F = N^{3/2} \sqrt{\frac{2 \pi^6}{27 \Vol(Y)}} + o(N^{3/2}) \,.
 }
Here, the volume of $Y$ is computed with an Einstein metric that satisfies the normalization condition $R_{mn} = 6 g_{mn}$.   The Freund-Rubin solution $AdS_4 \times Y$ arises as the near-horizon limit of a stack of $N$ M2-branes placed at the tip of the Calabi-Yau cone $X$ over $Y$.  The $N^{3/2}$ behavior of the number of degrees of freedom had been known for quite some time, as the same large $N$ dependence appears in other quantities such as the thermal free energy that was computed in \cite{Klebanov:1996un} more than ten years ago.  A field theory explanation of this peculiar large $N$ dependence had been lacking until recently, mostly because explicit Lagrangian descriptions of the field theories living on coincident M2-branes have been discovered only in the past few years.  Starting with ABJM theory \cite{Aharony:2008ug} that describes $N$ M2-branes sitting at an orbifold singularity of $\C^4$, there are now many Chern-Simons matter $U(N)$ gauge theories that are proposed to describe the effective dynamics on $N$ M2-branes placed at the tip of various Calabi-Yau cones (see for example \cite{Jafferis:2008qz, Franco:2009sp, Aganagic:2009zk, Jafferis:2009th, Benini:2009qs}).  However, only few of these dualities have been extensively tested.  Extensive tests are difficult to perform because supergravity on $AdS_4 \times Y$ is supposed to be a good approximation to the dynamics of the CS-matter gauge theories only as one takes the gauge group ranks to infinity while keeping the CS levels fixed.  In this limit, the 't Hooft coupling $N/k$ becomes large, and there are no perturbative computations that one can perform.

That in ${\cal N} \geq 2$ theories one can write $F$ exactly in terms of a matrix integral means that by evaluating this integral one can test some of these $AdS_4$/CFT$_3$ dualities that have been put forth in recent years.   In particular, one can provide a field theory derivation of the $F \propto N^{3/2}$ large $N$ dependence of the number of degrees of freedom \eqref{MtheoryExpectation} on $N$ coincident M2-branes.  Moreover, one can compare the coefficient of $N^{3/2}$ in \eqref{MtheoryExpectation} to the volume of the internal space $Y$ that one can compute independently by integrating the square root of the determinant of the Sasaki-Einstien metric on $Y$.  Such comparisons were made in \cite{Marino:2009jd} in the case of ABJM theory, in \cite{Herzog:2010hf, Santamaria:2010dm} for a large class of theories with ${\cal N} = 3$ supersymmetry, and in \cite{Martelli:2011qj, Cheon:2011th, Jafferis:2011zi} for many theories with ${\cal N} = 2$ supersymmetry.  In \cite{Marino:2009jd, Santamaria:2010dm} $F$ was computed as a function of the 't Hooft coupling, while the other works focused only on the strong coupling regime and used the method developed in \cite{Herzog:2010hf} to evaluate the multi-matrix integrals at large $N$ and fixed CS levels in a saddle point approximation.

In this paper we build upon the work in \cite{Herzog:2010hf,Jafferis:2011zi} in several ways.  Our first main result consists of an infinite class of RG flows whose IR and ultraviolet (UV) fixed points preserve
${\cal N} = 3$ SUSY.   Via the AdS/CFT correspondence, at large $N$ such a flow is dual to a
holographic RG flow between two $AdS_4 \times Y$ extrema of 11-d supergravity.
That the extrema preserve ${\cal N} = 3$ supersymmetry 
means that the cones over the spaces $Y$ are hyperk\"ahler and, because of that, the spaces $Y$ are called tri-Sasakian. (We will provide a more detailed exposition of hyperk\"ahler spaces at the beginning of section~\ref{TORIC}.)  The 11-d SUGRA solutions that represent M2-branes sitting at the tip of a hyperk\"ahler cone were constructed in \cite{Gauntlett:1997pk}, where they were also related through string duality to a brane construction in type IIB string theory.  The type IIB brane construction consists of $N$ D3-branes filling the $0126$ directions, and a sequence of $(p_a, q_a)$ five-branes,\footnote{We adopt the convention where by a $(p, q)$ five-brane we mean a brane with $p$ units of NS5 charge and $q$ units of D5 charge.} $1 \leq a \leq d$, filling the $012$ directions and sitting at fixed angles in the $37$-, $48$-, and $59$-planes.  The $(p,q)$ branes break the D3-brane stack into $d$ segments, and the three-dimensional $U(N)^d$ CS-matter gauge theories then live on the segments.   The ${\cal N} = 3$ RG flows we want to study correspond to removing one of the $(p, q)$ branes or to a $(p_1, q_1)$ brane and a $(p_2, q_2)$ brane combining into a $(p_1 + p_2, q_1 + q_2)$ brane.  We find that the $F$-theorem is satisfied in these examples.  

Our results for the $F$-theorem follow most naturally from 
the gravity side of the AdS/CFT correspondence which provides an $SL(2, \mathbb{R})$ invariant result for $F$ as a function of arbitrary $(p_a, q_a)$.
We provide field theory confirmation of the gravity result in the cases where $p_a=0$ or $p_a=1$ for each $a$. 
While the 3-d field theory is simplest when $p_a = 0$ or 1,
for $p_a > 1$, ref.\ \cite{Gaiotto:2008ak} provides a more complicated field theory 
description involving 
an interpolating $T(U(N))$ theory that implements a sort of local $S$ duality.  Instead of using the $T(U(N))$ theory, we derive the matrix model for $p_a>1$ by bootstrapping from our large $N$ results.  We check that the matrix model yields the correct answer in the large $N$ limit,
and also that it is invariant under $SL(2, \mathbb{Z})$.
As a bonus, we discover the matrix model of the $T(U(N))$ theory.\footnote{%
 As we were completing this work, \cite{Benvenuti:2011ga} appeared which contains this same result.  Also, M.~Yamazaki, T.~Nishioka and Y.~Tachikawa have informed us they have independently derived the $T(U(N))$ partition function \cite{Nishioka:2011dq}.
}
The matrix model for arbitrary $p_a$ is our second main result.

Along the way we tie several loose ends left off from \cite{Herzog:2010hf}.  In \cite{Herzog:2010hf} the matrix model was solved explicitly only for theories on $(1, q_a)$ branes with $a \leq 4$.  We provide a general solution that holds for any number of $(p_a, q_a)$ branes.  In \cite{Herzog:2010hf} it was checked numerically in a few examples that the M-theory prediction for the volume of $Y$ that can be extracted from \eqref{MtheoryExpectation} agrees with the geometric computation performed by Yee \cite{Yee:2006ba}.  We build on Yee's work and prove that in general the two computations of $\Vol(Y)$ agree.  In \cite{Herzog:2010hf} it was conjectured that the volume $\Vol(Y)$ can be expressed in terms of a certain sum over trees.  We provide a proof of this tree formula.

Our third main result, contained in Section \ref{FIELDTHEORY},  
is a relationship between the eigenvalues in the matrix model and the number of chiral operators in the field theory.
Define $\psi(r)$ to be the number of chiral operators with $R$-charge smaller than $r$ for the $N=1$ gauge theory.  
The authors of \cite{Bergman: 2001qi, Martelli:2006yb} demonstrated that there is a relationship between $\psi(r)$ and the volume of the Sasaki-Einstein space in the large $r$ limit.
Given \eqref{MtheoryExpectation}, there must also be a relationship between $\psi(r)$ and $F$.
In fact, as we show in this paper for the necklace theories,
more precise relationships can be established between the matrix model and operator counting problems. 
 The operators in the necklace theories also have a monopole charge $m$ corresponding to $m$ flux units in a diagonal subgroup of the gauge group.  
 Thus, we may consider $\psi(r,m)$ to be the number of operators with $R$-charge less than $r$ and monopole charge less than $m$.  Let $X_{ab}$ be a hypermultiplet transforming under the fundamental of the $b$th gauge group and the antifundamental of the $a$th gauge group.
 We can consider $\psi_{X_{ab}}(r,m)$ to be defined as above but now with the operator $X_{ab}=0$.  
In the large $N$ limit, the matrix model is solved by a saddle point approximation for which the eigenvalues are complex numbers $\lambda_a = N^{1/2} x + i y_a$.  
The eigenvalues can be parametrized by an eigenvalue density $\rho(x)$, which turns out to be the same for each gauge group, and an imaginary part $y_a(x)$.  
Our two results are
\begin{eqnarray}
\label{resultone}
\left. \frac{\partial^3 \psi}{\partial r^2 \partial m} \right|_{m = rx / \mu} &=& \frac{r}{\mu} \rho(x) \ , \\
\label{resulttwo}
\left. \frac{\partial^2 \psi_{X_{ab}}}{\partial r \partial m} \right|_{m=rx / \mu} &=& \frac{r}{\mu} \rho(x)
[ y_b(x) - y_a(x) + R(X_{ab})] \ ,
\end{eqnarray}
where $\mu = 3 F / 4 \pi N^{3/2}$.
(Here we take the liberty of replacing the operator counts, which are discrete
functions, with continuous approximations.)
We believe these relations will hold more generally.

\section{Volumes of Toric tri-Sasaki Einstein Spaces}

On the gravity side of the AdS/CFT duality, we have M-theory backgrounds generated by placing a stack of $N$ M2-branes at the tip of a hyperk\"ahler cone, with $N$ large.  In this section our aim is to compute the free energy of the M2-brane theory purely from the supergravity side of the correspondence.  We start by introducing in section~\ref{TORIC} the main ingredients in constructing the 11-d supergravity solution, namely the toric hyperk\"ahler spaces.  In section~\ref{VOLUME} we build upon the results of Yee \cite{Yee:2006ba} and express the volume of these spaces in terms of the volume of a certain polygon for which we will provide a field theory interpretation later on.  In section~\ref{BRANE} we comment on some field theory implications of this formula, and show explicitly that the free energy decreases along certain RG flows, in agreement with the $F$-theorem proposed in \cite{Jafferis:2011zi}.

\subsection{Toric Hyperk\"ahler Cones from a Quotient Construction}
\label{TORIC}

We start by introducing the toric hyperk\"ahler cones.  The following discussion draws heavily from \cite{BD} and \cite{Yee:2006ba}.

A hyperk\"ahler manifold possesses $4n$ real dimensions and has a Riemannian metric $g$ which is k\"ahler with respect to 3 anti-commuting complex structures $J_1$, $J_2$, and $J_3$.  These $J_i$ furthermore 
satisfy the quaternionic relations $J_1^2 = J_2^2 = J_3^2 = J_1 J_2 J_3 = -1$.   The simplest example of a hyperk\"ahler manifold is the four-dimensional space of quaternions $\HH \cong \R^4$, endowed with the standard line element.  A single quaternion $q \in \HH \cong \R^4$ can be represented as a two-by-two complex matrix
\be
q = \left(
\begin{array}{cc}
u & v \\
- \bar v & \bar u 
\end{array}
\right) \label{uvQuaternion}
\ee
parametrized in terms of two complex variables $u$ and $v$.  In terms of $u$ and $v$ the metric is $ds^2 = \abs{du}^2 + \abs{dv}^2$ and the three k\"ahler forms are
\begin{eqnarray}
\label{flat_triplet}
\omega^3 &=& - \frac{i}{2} (du \wedge d\bar u + dv \wedge d \bar v) \ , \\
(\omega^1 - i \omega^2) &=& i (du \wedge dv) \ . 
\end{eqnarray}
The $\omega^a$ transform as a triplet under the $SU(2)$ symmetry that acts as left multiplication on $q$.

The quaternions $\HH$ can also be written as a $U(1)$ bundle over $\R^3$ where the $U(1)$ fiber shrinks to zero size at the origin of $\R^3$.  This description comes from the uplift of the Hopf fibration from $S^3$ to $\R^4$ and makes explicit an $SO(3) \times U(1)$ subgroup of the $O(4)$ rotational symmetry of $\R^4$.  More explicitly, if one writes 
 \be
u = \sqrt{r} \cos \left( \theta/2 \right) e^{ - i \phi /2+i \psi} \; , \; \; \;
v = \sqrt{r} \sin \left( \theta/2 \right) e^{ - i \phi/2-i \psi} \ ,
\ee
then the standard line element on $\HH$ becomes
\begin{eqnarray}
ds^2 = 
du \, d\bar u + dv \, d \bar v 
= \frac{1}{4} \frac{d\vec{r}^2}{r}
 + r \left( d \psi - \frac{1}{2} \cos( \theta) d\phi \right)^2 \,,
 \label{R3Fibration}
\end{eqnarray}
where $\vec{r} = (r \sin \theta \cos\phi, r \sin \theta \sin \phi, r \cos \theta)$ is the usual parameterization of $\R^3$ in terms of spherical coordinates.  The coordinate $\psi \in [0, 2 \pi)$ parameterizes the $U(1)$ Hopf fiber.  From eq.~\eqref{uvQuaternion} we see that rotations in the $U(1)$ fiber correspond to phase rotations of $q$.

Another example of a hyperk\"ahler manifold is the Cartesian product of $d$-copies of the quaternions ${\mathbb H}^d \cong {\mathbb R}^{4d}$, also considered with the flat metric
 \es{MetricHd}{
  ds^2 = \sum_{a = 1}^d \left[ \frac{1}{4} \frac{d\vec{r}_a^2}{r_a}
 + r_a \left( d \psi_a - \frac{1}{2} \cos( \theta_a) d\phi_a \right)^2 \right]\,.
 }
Starting with $\HH^d$, one can construct construct a large number of further examples of toric hyperk\"ahler spaces using the hyperk\"ahler quotient procedure of Hitchin, Karlhede, Lindstr\"om, and Ro\v{c}ek \cite{HKLR}.  A toric hyperk\"ahler manifold can be defined to be a hyperk\"ahler quotient of ${\mathbb H}^d$ for some integer $d$ by a $(d-n)$-dimensional subtorus $N \subset T^d$.   The underlying reason that the quotient continues to be hyperk\"ahler is that the triplet of K\"ahler forms (\ref{flat_triplet}) is invariant under the $U(1)$ action that sends $u \to u e^{i \alpha}$ and $v \to v e^{-i \alpha}$.  

The hyperk\"ahler quotient procedure requires the data of how $N = T^{d-n}$ sits inside $T^d$.  The inclusion $N \subset T^d$ can be described by the short exact sequence of tori
 \es{ToriSE}{
  0 \longrightarrow T^{d-n} \stackrel{i}\longrightarrow T^d \stackrel{\pi}\longrightarrow T^n \longrightarrow 0 \,,
 }
where we also introduced the quotient $T^n \cong T^d / T^{d-n}$, as well as the inclusion map $i$ and the projection map $\pi$.  Each torus $T^k$ in this sequence can be identified with $\R^k / (2 \pi \Z^k)$, so the data of how $N$ sits inside $T^d$ can be encoded in how the standard lattice $\Z^{d-n} \subset \R^{d-n}$ sits inside $\Z^d \subset \R^d$, or equivalently how $\Z^d \subset \R^d$ projects down to $\Z^n \subset \R^n$.  Such a construction can be described by a short exact sequence of vector spaces
 \es{RSE}{
  0 \longrightarrow \R^{d-n} \stackrel{Q}{\longrightarrow} \R^d \stackrel{\beta}{\longrightarrow} \R^n \longrightarrow 0 \,,
 }
which restricts to a short exact sequence of free-$\Z$ modules (lattices):
 \es{ZSE}{
   0 \longrightarrow \Z^{d-n} \stackrel{Q}{\longrightarrow} \Z^d \stackrel{\beta}{\longrightarrow} \Z^n \longrightarrow 0 \,.
 }
Since $Q$ and $\beta$ are linear maps, they can be represented by matrices:  $Q$ by a $d \times (d-n)$ matrix and $\beta$ by an $n \times d$ matrix.  That the sequence \eqref{RSE} is exact means that precisely $n$ columns of $\beta$ are linearly independent (i.e.~$\beta$ is surjective), that $\beta Q = 0$, and that the $d-n$ columns of $Q$ are linearly independent (i.e.~$Q$ is injective).  That \eqref{RSE} restricts to \eqref{ZSE} further implies that $Q$ and $\beta$ have integer entries.  The torus $N = \R^{d-n} / (2 \pi \Z^{d-n})$ as we defined it has volume $(2 \pi)^{d-n}$.

The hyperk\"ahler quotient of ${\mathbb H}^d$ by the torus $N$ is defined to be the zero locus of a set of moment maps in addition to a quotient by the torus action.  The $3(d-n)$ moment maps are compactly expressed using the matrix equation
\begin{eqnarray}
\mu_i &=& \sum_{a=1}^d Q^a_i  \left[ q_a \sigma_3 q_a^\dagger  + 
\vec \lambda_a \cdot \vec \sigma
\right]
\\
&=&
\sum_{a=1}^d  Q^a_i
\left[ \left(
\begin{array}{cc}
r_{a3} & r_{a1} - i r_{a2}
 \\
r_{a1} + i r_{a2}  & -r_{a3}
\end{array}
\right)
+
\left(
\begin{array}{cc}
\lambda_{a3} & \lambda_{a1} - i \lambda_{a2} \\
\lambda_{a1} + i \lambda_{a2} & - \lambda_{a3} 
\end{array}
\right)
\right]
 \,,
\end{eqnarray}
where $\sigma_1$, $\sigma_2$, and $\sigma_3$ are the Pauli spin matrices.
The hyperk\"ahler quotient $X$ is then
\be
X \equiv {\mathbb H}^d /// N = \mu^{-1}(0) / N \ .
\ee

We are particularly interested in the case where $\vec \lambda_a = 0$ and the corresponding hyperk\"ahler manifold is a cone.  The base of this cone is a $4n-1$ real dimensional Riemannian manifold with positive scalar curvature and a locally free action of $SU(2)$ that descends from the $SU(2)$ rotating the three complex structures.  The base of such a cone is called a tri-Sasaki manifold.  Indeed, a Riemannian manifold $(Y,g)$ is tri-Sasaki if and only if the Riemannian cone $X = ({\mathbb R}^+ \times Y, dr^2 + r^2 \, g_{ij} dy^i dy^j)$ is hyperk\"ahler.  The induced metric on $X$ from ${\mathbb H}^d$ is Ricci flat, which in turn implies that the induced metric on $Y$ satisfies Einstein's equations with a positive cosmological constant.  In other words, $Y$ is also an Einstein manifold.

\subsection{The Volume of the tri-Sasaki Einstein Base}
\label{VOLUME}

Given a toric hyperk\"ahler cone $X$, we would like to compute the volume of the base $Y$ with respect to the induced metric from ${\mathbb H}^d$.  A brute force approach would be to introduce the variables $\vec t_j$ such that
\be
\vec r_a = \sum_j \beta_{aj} \vec t_j \ .
\ee
These constrained $\vec r_a$ automatically satisfy the moment map conditions.  
Plugging this expression into the line element \eqref{MetricHd} on ${\mathbb H}^d$ yields a metric on
$X \times T^{d-n}$.  Integrating the square root of the determinant of the metric up to a finite radial coordinate $1 = \sum_a r_a$ and dividing by $\Vol(T^{d-n})$ would yield $\Vol(Y)/4n$.

The approach outlined above appears to be difficult.  Instead, we will use a result due to Yee \cite{Yee:2006ba}, which was proven using an elegant localization argument:
\begin{YF}
Consider a tri-Sasaki Einstein manifold $Y$ defined via the short exact sequence (\ref{ZSE}) and hyperk\"ahler quotient construction described above.  If the metric is normalized such that
$R_{ij} = 2(2n-1) g_{ij}$, then the volume is
\be
\Vol(Y) = 
\frac{2^{d-n+1} \pi^{2n}}{(2n-1)! \Vol\left(N\right)}  \int \left( \prod_{j=1}^{d-n} d{ \phi^j} \right) \prod_{a=1}^{d} \frac{1}{1 + \left( \sum_{k=1}^{d-n} Q_k^a \phi^k \right)^2} \,.
\label{Yeesresult}
\ee
\end{YF}
This volume formula has the forgiving property of being invariant under rescaling the matrix
$Q \to \lambda Q$.  Under this rescaling the torus volume changes, $\Vol(N) \to \lambda^{n-d} \Vol(N)$.  The factor $\lambda^{n-d}$ is then canceled by the Jacobian introduced upon rescaling $\phi \to \phi / \lambda$.  
We discussed above that the columns of $Q$ can be chosen such they form a ${\mathbb Z}$-basis of the kernel of $\beta$ and $\Vol(N) = (2 \pi)^{d-n}$.  Note that if we were not so clever in our choice of $Q$, the volume formula would still give us the correct answer because of this scaling invariance. 

So far we have assumed that $N$ is isomorphic to $T^d$, but we can also
choose $N$ to be isomorphic to $N_T \times T^d$ for some finite 
abelian group $N_T$.
In that case, the image of $\beta$ is no longer all of $\Z^n$.
The cokernel of $\beta$ is dual to $N_T$.
The volume of $N$ for our choice of $Q$ is therefore
$(2 \pi)^{d-n} \abs{N_T} = (2 \pi)^{d-n} \abs{\coker \beta}$.

We now rewrite this volume integral in a more convenient fashion.
Let us assume that we chose the 
columns $Q$ to form a ${\mathbb Z}$ basis of $\mbox{Ker}(\beta)$. 
 In this case, the volume formula simplifies:
\begin{equation} \label{eq:vol}
\Vol(Y) = \frac{2 \pi^{3n-d}}{(2n-1)!}  \int \left( \prod_{j=1}^{d-n} d{\phi^j} \right) \prod_{a=1}^d 
\frac{1}{1+\left( \sum_k Q^a_k \phi^k \right)^2 }.
\end{equation}
A well-known result from Fourier analysis is
\begin{equation}
\frac{1}{1+x^2} = \frac{1}{2} \int_{-\infty}^{\infty} d{y} e^{-|y|+ixy}.
\end{equation}
Then $\Vol(Y)$ can be rewritten as
\begin{equation} \label{eq:volboth}
\Vol(Y) = \frac{\pi^{3n-d}}{(2n-1)! \cdot 2^{d-1}}
\int \left( \prod_{j=1}^{d-n} d{\phi^j} \right) \left( \prod_{a=1}^d  e^{-|y_a|} d{y_a} \right) \exp\left(\sum_{k,a} i Q^a_k \phi^k y_a \right).
\end{equation}
We are allowed to switch the order of integration, integrating over the $\phi^j$ first.\footnote{The integral \eqref{eq:volboth} is not absolutely convergent.
However,
if we multiply the integrand by $\exp \left[ -\sum_k \epsilon_k (\phi^k)^2 \right]$, for some small $\epsilon_k > 0$, then the
integral does converge absolutely.  We can then take the limit
$\epsilon_k \to 0^{+}$.
If we integrate out the $y_a$ then we get
$\lim_{\epsilon_k \to 0^+} \frac{2 \pi^{3n-d}}{(2n-1)!}  \int \left( \prod_{j=1}^{d-n} d{\phi^j} \right) \prod_{a=1}^d \frac{\exp \left[ -\sum_k \epsilon_k (\phi^k)^2 \right]}{1+\left( \sum_k Q^a_k \phi^k \right)^2 } = \frac{2 \pi^{3n-d}}{(2n-1)!}  \int \left( \prod_{j=1}^{d-n} d{\phi^j} \right) \prod_{a=1}^d \frac{1}{1+\left( \sum_k Q^a_k \phi^k \right)^2 }$.
If we integrate out the $\phi^j$ then we get
$\lim_{\epsilon_k \to 0^+} \frac{\pi^{2n}}{(2n-1)! \cdot 2^{n-1}} \int \left( \prod_{a=1}^d e^{-|y_a|} d{y_a} \right) \prod_{k=1}^{d-n} \frac{\exp -\left(\sum_a Q^a_k y_a\right)^2/4\epsilon_k}{2 \sqrt{\pi \epsilon_k}}= \frac{\pi^{2n}}{(2n-1)! \cdot 2^{n-1}} \int \left( \prod_{a=1}^d e^{-|y_a|} d{y_a} \right) \prod_{k=1}^{d-n} \delta \left(\sum_a Q^a_k y_a\right)$. 
}
We obtain 
\be \label{eq:volexp}
\Vol(Y)  =  \frac{\pi^{2n}}{(2n-1)! \cdot 2^{n-1}} \int \left( \prod_{a=1}^d e^{-|y_a|} d{y_a} \right) \prod_{k=1}^{d-n} \delta \left(\sum_a Q^a_k y_a\right)  \ .
\ee
To integrate over the delta functions, note that we can get a basis for ${\mathbb Z}^d$ by taking the basis for $\ker(\beta)$ and pullbacks of the basis for $\im(\beta)$.  The Jacobian for transforming from the standard basis of ${\mathbb Z}^p$ to this basis must be one, since both bases generate the same lattice.  In our new coordinates, we have two kinds of variables: $s_i$ corresponding to the columns of $Q$ and $t_i$ corresponding to the rows of $\beta$.    The product of delta functions is just $\delta(s_1) \delta(s_2) \cdots \delta(s_{d-n})$ and can be performed straightforwardly.  
If the rows of $\beta$ span ${\mathbb Z}^n$, the $t_j$ can be written
\be
y_a = \sum_{j=1}^{n} \beta_{aj} t_j \ .
\label{practicaltj}
\ee
The integral reduces to the following useful form:
\be
\Vol(Y) =   \frac{\pi^{2n}}{(2n-1)! \cdot 2^{n-1}} \int \left( \prod_{k=1}^n d{t_k} \right) \exp \left(-\sum_{a=1}^d \left| \sum_{j=1}^n \beta_{aj} t_j \right| \right) \,.
\label{usefulVol}
\ee
Note that if the rows of $\beta$ do not span ${\mathbb Z}^n$, then the Jacobian will have an extra factor of 
$|\coker(\beta)|$, which cancels the $|\coker(\beta)|$ in $\Vol(N)$.  So
\eqref{usefulVol} holds even if the rows of $\beta$ do not span $\Z^n$.
This integral form provides us with a corollary to Yee's Formula:
\begin{cor}
i) If a column $\beta_a$ is removed from $\beta$, $\Vol(Y)$ increases.
ii) If two columns $\beta_a$ and $\beta_b$ of $\beta$ are combined to form the new column $\beta_a+\beta_b$ of a new $\beta'$ with one fewer columns, $\Vol(Y)$ either increases or stays the same.  
The volume remains the same if and only if the two columns are proportional, $\beta_a = c \beta_b$ for some $c \in {\mathbb R}^+$. 
\label{cor:RG}
 \end{cor}
  \begin{proof}
 (i) is true 
 because the integrand of (\ref{usefulVol}) (which is positive definite) increases when a column of $\beta$ is removed.  
 (ii) is true because the absolute value of a sum is less than or equal to the sum of the absolute values.  Equality occurs only when the expressions inside the 
 absolute values are proportional.
\end{proof}

The volume formula (\ref{usefulVol}) leads to another useful corollary of Yee's Formula:
\begin{cor}
Let ${\cal P} \in \R^n$ be the polytope
 \es{Polytope}{
  {\cal P} = \left\{ \vec{t} \in \R^n :  \sum_{a=1}^d \left| \sum_{j=1}^n \beta_{aj} t_j \right| \leq 1 \right\} \,.
 }
 Let $S^{4n-1}$ be a $(4n-1)$-dimensional sphere with unit radius.  The volume of the tri-Sasaki Einstein manifold satisfies the relation
 \es{VolYVolSphere}{
  \frac{\Vol(Y)}{\Vol(S^{4n-1})} = \frac{n!}{2^n} \Vol({\cal P}) \,.
 }
 \label{cor:polygon}
\end{cor}
\begin{proof}
Introduce the notation $\alpha = \sum_{a=1}^d \left| \sum_{j=1}^n \beta_{aj} t_j \right|$ with $\alpha > 0$ if at least one of the $t_k$ is non-zero, and write $t_k = \alpha b_k$ for some new variables $b_k$.  The $b_k$ are constrained to live on the boundary of ${\cal P}$.  Eq.~\eqref{usefulVol} can be rewritten as
 \es{usefulVolRewrittten}{
  \Vol(Y)= \frac{\pi^{2n}}{(2n-1)! \cdot 2^{n-1}} \int \left( \prod_{k=1}^n d{b_k} \right) \delta \left(\sum_{a=1}^d \left| \sum_{j=1}^n \beta_{aj} b_j \right| -1 \right) \int_0^\infty d \alpha\, \alpha^{n-1} e^{-\alpha} \,.
 }
The volume of ${\cal P}$ can be also written in terms of these variables:
 \es{VolP}{
  \Vol({\cal P}) = \int \left( \prod_{k=1}^n d{b_k} \right) \delta \left(\sum_{a=1}^d \left| \sum_{j=1}^n \beta_{aj} b_j \right| -1 \right) \int_0^1 d \alpha\, \alpha^{n-1} \,.
 }
Performing the integrals in $\alpha$ and comparing the last two formulas, we obtain
 \es{VolYVolP}{
  \Vol(Y) = \frac{\pi^{2n} n!}{(2n-1)! \cdot 2^{n-1}} \Vol({\cal P}) \,,
 }
or, using the fact that the volume of the $(4n-1)$-sphere is $\Vol(S^{4n-1}) = 2 \pi^{2 n} / (2n-1)!$
the desired result follows.
\end{proof}

A final corollary is an explicit result for the volume of the tri-Sasaki Einstein spaces relevant for the gauge theories we discuss below:
\begin{cor}
In the case $n=2$, 
choose $\beta$ such that the two-vectors $\beta_a$ lie in the upper half plane and order them 
such that $\beta_a \wedge \beta_{a+1} >0$.\footnote{We define $\begin{pmatrix}a\\ b \end{pmatrix} \wedge \begin{pmatrix}c\\ d \end{pmatrix} \equiv ad - bc$.} 
The volume of the tri-Sasaki Einstein space $Y$ is
\begin{equation}
\Vol(Y) = \frac{\pi^4}{6} \sum_{a=1}^d \frac{\gamma_{a(a+1)}}{\s_a \s_{a+1}} 
\end{equation}
where we have defined the quantities 
$\gamma_{ab} \equiv | \beta_a \wedge \beta_b|$,  
$
\s_a \equiv \sum_{b=1}^d  \gamma_{ab}
$,
and $\beta_{d+1} \equiv - \beta_1$.
\label{cor:n2ordered}
\end{cor}

Note that the volume of $Y$ is  
independent of the sign of the $\beta_a$ and their order inside $\beta$.
The order of the $\beta_a$ in the corollary is their order around
$\mathbb{R}\mathbb{P}^1$.  (Flipping the sign of any given column of $\beta$ does not change the hyperk\"ahler quotient $\mathbb{H}///N$ because the torus $N$ defined by \eqref{ZSE} is invariant under such a sign flip.  Therefore it is natural to identify $\beta_a$ with $-\beta_a$ and consider the ordering of the $\beta_a =(p_a, q_a)^T$ to be the ordering of $[p_a, q_a] \in \mathbb{R}\mathbb{P}^1$ around $\mathbb{R}\mathbb{P}^1$.)

\begin{proof}[Proof of Corollary \ref{cor:n2ordered}]
 We split the integral into various regions according to
whether $\sum_j \beta_{aj} t_j$ is positive or negative.  There are $2d$
such regions.  Since the integral does not change if we replace each $t_k$ with
$-t_k$, we only need to consider $d$ regions.
The ordering of the $\beta_a$ guarantees that any two consecutive columns determine a region boundary.

Now we choose $a$ and compute the integral in a region bounded by the lines
$u = -\sum_j \beta_{aj} t_j = 0$ and $v = \sum_j \beta_{(a+1)j} t_j = 0$.  The prefactor in the integral is $\pi^4/12$, but we need to multiply by two since there are two
such regions.
We get
\begin{eqnarray}
& & \frac{\pi^4}{6}\frac{1}{ \gamma_{a(a+1)}} \int_0^\infty d{u} \int_0^\infty d{v} \exp \left (-\sum_{b=1}^p \frac{\gamma_{ab} v + \gamma_{(a+1)b} u}{\gamma_{a(a+1)} } \right)
\nonumber
 \\
& = & \frac{\pi^4}{6} \frac{\gamma_{a(a+1)}}{ \left(\sum_b \gamma_{ab}\right)\left(\sum_b \gamma_{(a+1)b} \right)}.
\end{eqnarray}
Summing the regions yields the volume stated in the corollary.
\end{proof}
Given Corollary \ref{cor:polygon}, there exists an equivalent proof that involves computing the area of ${\mathcal P}$ from the definition (\ref{Polytope}).

The authors of \cite{Herzog:2010hf} conjectured a formula for  the volume of these $n=2$ tri-Sasaki Einstein spaces.  Their formula is interesting because it does not rely on an ordering of the $\beta_a$ and makes the permutation symmetry of the columns of $\beta$ manifest.  We have been able to promote this conjecture to a theorem.  As the techniques for the proof are not typical of the main arguments in the paper, we include the proof as Appendix \ref{app:treeproof}.
\begin{TF}
The area of the polytope ${\mathcal P}$ described in (\ref{Polytope}) in the case $n=2$ can be written
\be
\mbox{\rm Area}({\mathcal P}) = 2 \frac{\sum_{(V,E)\in T} \prod_{(a,b)\in E} \gamma_{ab}}{\prod_{a=1}^d \sigma_a} \ .
\ee
where $T$ is the set of all trees (acyclic connected graphs) with nodes $V=\{ 1,2, \ldots, p \}$
and edges
$E = \{(a_1, b_1), (a_2,b_2), \ldots, (a_{p-1}, b_{p-1}) \}$.
\end{TF}

\subsection{Brane Constructions and an $F$-theorem}
\label{BRANE}

Consider the following brane construction in type IIB string theory.  A stack of $N$ coincident  D3-branes spans the $01236$ directions with the 6 direction periodically identified.  Let there be bound states of NS5- and D5-branes.  Denote the number of NS5- and D5-branes in the bound state by $p_a$ and $q_a$ respectively.   These $(p_a, q_a)$-branes intersect the D3-branes at intervals around the circle and span the 012 directions.  Each $(p_a, q_a)$-brane lies at an angle $\theta_a$ in the 37, 48, and 59 planes 
where $\theta_a = \arg (p_a + i q_a)$.
\begin {figure} [!t]
  \centering
\newcommand {\svgwidth} {0.7\textwidth}
 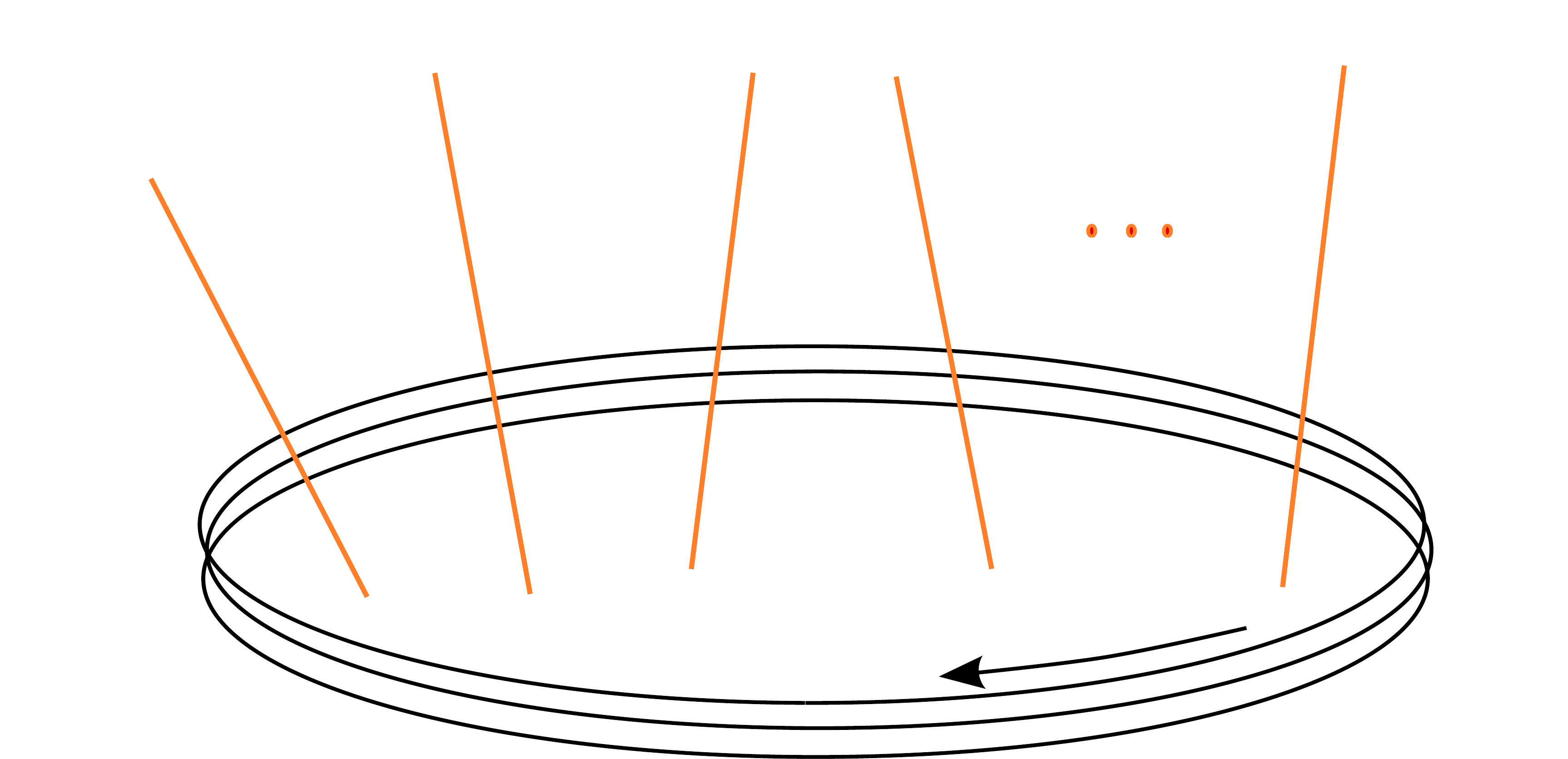
  \caption {A schematic picture of the brane construction.  The $N$ D3-branes span the $0126$ direction, and the $(p_a, q_a)$ 5-branes span the $012$ directions as well as the lines in the $37$, $48$, and $59$ planes that make angles $\theta_a = \arg (p_a + i q_a)$ with the $3$, $4$, and $5$ axes, respectively.  The three-dimensional ${\cal N}=3$ theories considered in this paper live on the $012$ intersection of these branes.\label{BranePicture}}
\end {figure}
These brane constructions are known to preserve 6 of the 32 supersymmetries of type IIB string theory \cite{Kitao:1998mf,Bergman:1999na}.\footnote{%
Note that a brane with charges $(-p_a, -q_a)$ is an anti-$(p_q,q_a)$-brane
rotated by 180 degrees, which is the same as a $(p_a, q_a)$-brane.
The $(p_a, q_a)$ charge is most naturally defined up to this overall sign.}

This brane construction cannot be described reliably within type IIB supergravity because the dilaton becomes large.  A better description can be obtained after a T-duality along the $6$ direction and a lift to M-theory, where the resulting configuration can be described within 11-d supergravity.   The geometry depends on $N$.  When $N$ is small, the geometry is ${\mathbb R}^{2,1} \times X$ where $X$ is the hyperk\"ahler cone discussed above for $n=2$ \cite{Gauntlett:1997pk}.
In the large $N$ limit, the D3-branes produce a significant back-reaction, and close to them the geometry is $AdS_4 \times Y$ where $Y$ is a tri-Sasaki Einstein space \cite{Jafferis:2008qz}.
In both cases, the charges $(p_a, q_a)=\beta_a^T$ are the columns of $\beta$.

In the case where $p_a = 0$ or 1, this brane construction admits a simple ${\mathcal N}=3$ supersymmetric 2+1 dimensional field theory interpretation.  The $(1,q_a)$-branes, 
$a=1, \ldots, d$, break the D3-branes up into $d$ segments along the circle.  For each segment, we have a $U(N)$ Chern-Simons theory at level $k_a = q_{a+1} - q_a$.  
The $(0,q_i)$-branes, $i = 1, \ldots, n_F$, also intersect the D3-branes.  At each intersection, we have massless strings that join the D3-branes to the $(0,q_i)$-brane that correspond to flavor fields in the fundamental representation.  (See figure~\ref{NecklaceFigure}.)
These gauge theories are described in greater detail in \cite{Imamura:2008nn,Jafferis:2008qz}.

As described in the introduction, it has been conjectured that the logarithm $F$ of the partition function of the Euclidean field theory on $S^3$ serves as a measure of the number of degrees of freedom in the theory, being an analog of the conformal anomaly coefficient $a$ of a 3+1 dimensional field theory.  
While the $a$-theorem is a conjecture that the conformal anomaly $a$ decreases along RG flows, the conjectured $F$-theorem states that $F$ should decrease along RG flows \cite{Jafferis:2011zi}.  
As given in (\ref{MtheoryExpectation}), the AdS/CFT correspondence predicts that in the large $N$ limit, $F \sim 1 / \sqrt{\Vol(Y)}$.  Thus, along RG flows, $\Vol(Y)$ should increase.  

It is remarkable that for the class of theories we consider here, Corollary \ref{cor:RG} confirms this expectation that $\Vol(Y)$ should increase along RG flows.  
The UV of the field theory should correspond to looking at our D3-brane construction from a great distance as only high energy excitations will be able to get far from the branes, while the IR should correspond to getting very close to the D3-branes.  
The simplest realization of such an RG flow is to add a mass term for flavor fields corresponding to a $(0,q_i)$-brane.  Such a mass term corresponds to introducing a small distance between the $(0,q_i)$-brane and the stack of D3-branes.  In the UV, these flavor fields will contribute to $F$, while in the IR, the flavors and corresponding $(0,q_i)$-brane should be absent.  

A slightly more complicated realization involves an 
intersecting $(0,q)$-brane and a $(1,0)$-brane.  
We have a $U(N)$ vector multiplet on each side of the $(1,0)$-brane.  
The $(0,q)$-brane produces
$q$ fundamental flavor fields charged under one $U(N)$ and $q$ anti-fundamentals charged under the other.  Adding real masses of the same sign for each flavor field corresponds to a
``web deformation'' \cite{Bergman:1999na} where close to the D3-branes the $(0,q)$- and $(1,0)$-brane form a $(1,q)$-brane bound state.  Far from the D3-branes, because we have fixed boundary conditions for the $(0,q)$- and $(1,0)$-brane by specifying their angles, the two branes will remain separate.
From a large distance, we will not see that the two 5-branes have combined.  We will only see their asymptotic regions where they remain separated.  Closer up, we will see that the branes have made a bound state.  From the point of view of our volume formula, 
combining these two branes
increases the volume and thus decreases $F$.

In fact, Corollary \ref{cor:RG} and our volume formula put no restriction on the type of $(p,q)$-brane we remove or the type of $(p,q)$- and $(p',q')$-brane we combine to form a $(p+p', q+q')$ bound state.  In all cases, the volume will increase, corresponding to a decrease in $F$.  
Unfortunately, from the field theory perspective it is not clear in general to what these more general types of RG flow correspond.

\section{Field Theory Computation of the Free Energy}
\label{FIELDTHEORY}

\subsection{${\mathcal N}=3$ Matrix Model}
\label{MATRIX}

Consider the ${\mathcal N}=3$ supersymmetric 2+1 dimensional Chern-Simons theories corresponding to the $(p, q)$-brane constructions described above where $p=0$ or 1.
Let there be $d$ $U(N)$ gauge groups at level $k_a = q_{a+1} - q_a$, matter fields $A_a$ and $B_a$ in conjugate bifundamental representations of the $(a-1)$st and $a$th gauge group and $n_a$ pairs of flavor fields transforming in fundamental and anti-fundamental representations of the $a$th gauge group, with $a=1, \ldots, d$ (see figure~\ref{NecklaceFigure}).  
\begin {figure} [!t]
  \centering
\newcommand {\svgwidth} {0.5\textwidth}
 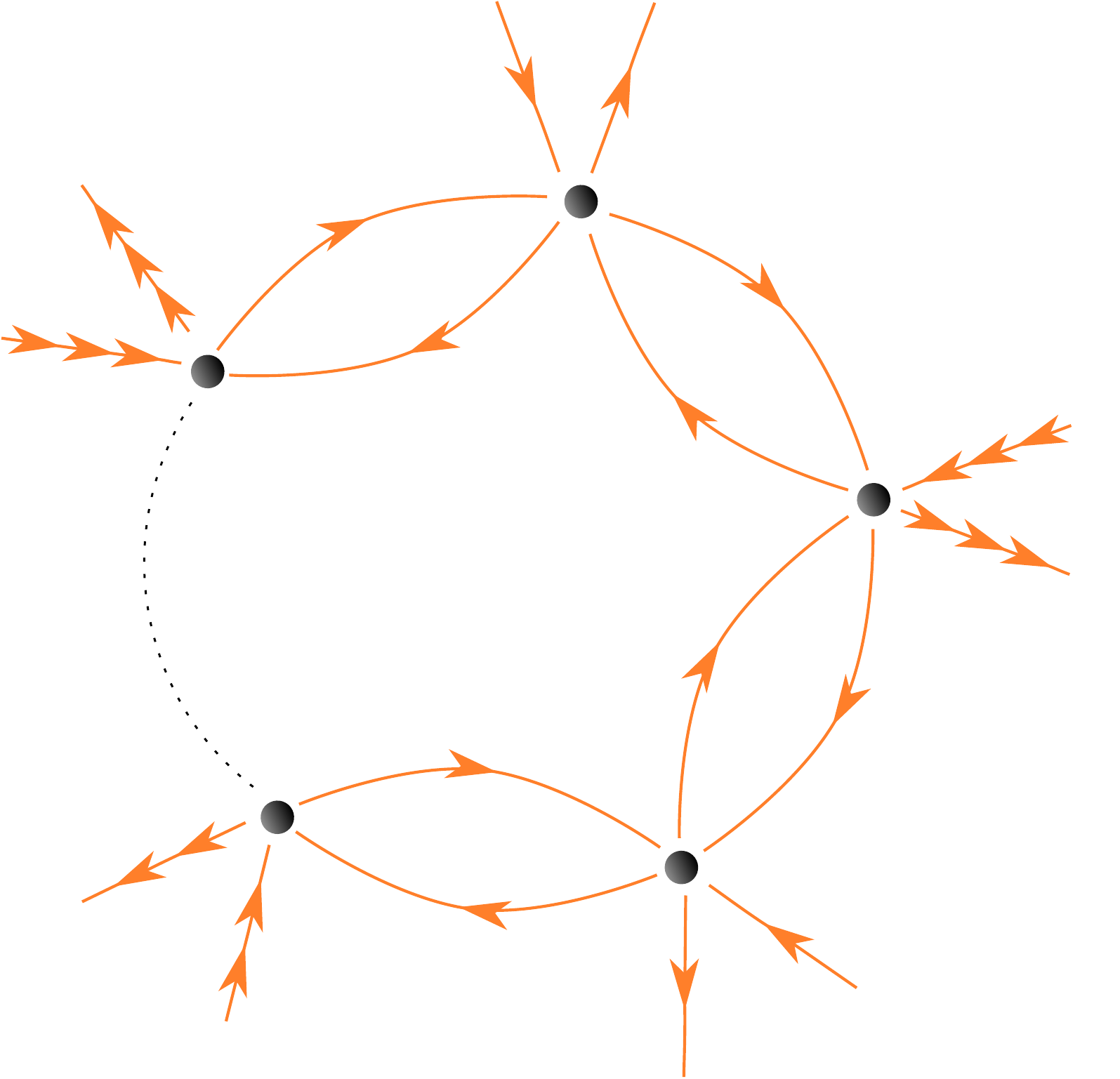
  \caption {A necklace quiver gauge theory where the gauge sector consists of $d$ $U(N)$ gauge groups with Chern-Simons coefficients $k_a$.  The matter content consists of the bifundamental fields $A_a$ and $B_a$, as well as $n_a$ pairs of fundamentals and anti-fundamentals transforming under the $a$th gauge group. 
  \label{NecklaceFigure}}
\end {figure}  
As explained in \cite{Kapustin:2009kz}, the partition function for these necklace quivers 
localizes on configurations where the scalars $\sigma_a$ in the vector multiplets are constant Hermitian matrices.  Denoting  the eigenvalues of $\sigma_a$ by $\lambda_{a, i}$, $1 \leq i \leq N$, the partition function takes the form of the eigenvalue integral
\es{NecklaceModel}{
Z = \int \left(\prod_{a,i} d \lambda_{a,i}  \right)
	L_v(\{ \lambda_{a,i} \} ) L_m(\{ \lambda_{a,i} \})
 }
where the vector multiplets contribute
\es{Lv}{
    L_v =  \frac{1}{N!} \prod_{a=1}^d \left( \prod_{i \neq j}
     2  \sinh [\pi (\lambda_{a,i} - \lambda_{a,j}) ]  \right)
   \exp \left(i \pi   \sum_{i}  k_a \lambda_{a,i}^2  \right)
}
and the bifundamental and fundamental matter fields contribute
\es{Lm}{
L_m = 
\prod_{a=1}^d \left( \prod_{i,j} \frac{1}{2 \cosh [\pi (\lambda_{a-1,i} - \lambda_{a,j})]  }\right)
\left( \prod_i \frac{1}{2 \cosh \pi \lambda_{a,i} } \right)^{n_a} \ .
}

We follow the recipe suggested in \cite{Herzog:2010hf} for analyzing this matrix model in the large $N$ limit.  We write $\lambda = N^{1/2} x + i y$ and assume that the density of eigenvalues $\rho(x)$ is the same for each vector multiplet.
To leading order in $N$, the matrix model for the ${\mathcal N}=3$ necklace theories involves extremizing a free energy functional of the type
\be
\label{necklaceF}
 F[\rho, \delta y_a] = \pi N^{3/2}  \int \rho(x) dx\, \left[ n_F | x | +
 2 x  \sum_{a=1}^d
     q_a  \delta y_a(x) +
     \rho(x) \sum_{a=1}^d f( \delta y_a(x)) \right]
      \,,
\ee
where $\delta y_a = y_{a-1} - y_a$, $n_F = \sum_a n_a$, and $f$ is
a periodic function with period one given by
\be
  f(t) = \frac{1}{4} - t^2 \qquad \text{when} \quad -\frac{1}{2} \leq t \leq \frac{1}{2} \,.
 \ee
This free energy should be extremized over the set
\be
  {\cal C} =  \left\{ (\rho, \delta y_a): \int dx\, \rho(x) = 1; \rho(x) \geq 0 \text{ and } \sum_{a=1}^d \rho(x) \delta y_a(x) = 0 \text{ a.e.} \right\} \,,
 \ee
where we think of $\rho(x)$ and $\rho(x) \delta y_a(x)$ as functions defined almost everywhere (a.e.).  To enforce these constraints, we introduce the Lagrange multipliers $\mu$ and $\nu(x)$:
\be
    \tilde F[\rho, \delta y_a] = F[\rho, \delta y_a] - 2 \pi N^{3/2} \left[ \mu \left( \int dx\, \rho(x) - 1 \right)
     +  \int dx\, \rho(x) \nu(x) \sum_{a = 1}^d \delta y_a(x) \right] \ .
\ee
The equations of motion that follow from this action are
\begin {subequations}
\label{rhozaeqs}
\begin{align}
\sum_{a=1}^d  \left[  f(\delta y_a(x)) \rho(x) + \left(q_a x-\nu(x)\right) \, \delta y_a(x)  \right] &= \mu - \frac{1}{2}  n_F |x|\,,
\label{rhoeq} \\
\frac{1}{2} f'(\delta y_a(x)) \rho(x)  + q_a x &= \nu(x) \,.
\label{zaeq}
\end{align}
\end {subequations}

The solution of these equations and the constraint $\sum_{a = 1}^d \delta y_a = 0$ is 
 \es{Soln}{
  \rho(x) &= s_L(x) - s_S(x) \,, \qquad 
   \nu(x) = - \frac{1}{2} \left[ s_L(x) + s_S(x)\right]  \,, \\
   \delta y_a(x) &= \frac{1}{2} \frac{\abs{s_L(x) + q_a x }  
     -\abs{s_S(x) + q_a x} }{s_L(x) - s_S(x)}  \,,
 }
where we denoted by $s_L(x)$ and $s_S(x)$ the two solutions of the equation
 \es{sEqSimp}{
 \frac{1}{2} n_F |x| +  \frac{1}{2} \sum_{a = 1}^d \abs{s(x) + q_a x} = \mu  \,,
 }
with $s_L(x) \geq s_S(x)$.
Note that the set of $s$ and $x$ satisfying (\ref{sEqSimp}) defines a polygon.  In fact, it defines the polygon ${\mathcal P}$ of (\ref{Polytope}) but rescaled by a factor of $2\mu $.

We show in Appendix~\ref{VIRIAL} quite generally that in the continuum limit  the extremized value of the free energy $F$ is proportional to the Lagrange multipler $\mu$:
 \es{FmuRelationText}{
  F = \frac{4 \pi N^{3/2}}{3} \mu \,.
 }
Thus to determine $F$, it suffices to find $\mu$.

By definition, the density $\rho(x)$ should integrate to one.
The solution (\ref{Soln}) demonstrates that the density $\rho(x)$ is proportional to the length of a slice through ${\mathcal P}$ 
at constant $x$.  Thus integrating $\rho(x)$ over $x$ should yield a quantity proportional to the area of the polygon.
We obtain
 \es{rhoNorm}{
  1 = \int dx\, \rho(x) = 4 \mu^2  \Vol({\cal P}) \,.
 }
 Assembling (\ref{FmuRelationText}) and (\ref{rhoNorm}) yields
 $
 F =2 \pi N^{3/2} / 3 \sqrt{\Vol({\cal P})}
$.
From Corollary \ref{cor:polygon} and in particular (\ref{VolYVolP}), we recover our AdS/CFT prediction (\ref{MtheoryExpectation}). 

\subsection{Operator Counting and the Matrix Model}

In this section, we relate the matrix model 
quantities $\rho(x)$ and $\rho(x) \delta y_a(x)$ to numbers of operators in the chiral ring that don't vanish on the geometric branch of the moduli space of the $\mathcal{N}=3$ necklace quiver in the abelian case $N=1$. 
To that end, we first characterize the chiral ring,
 defined to be gauge-invariant combinations 
 of the bifundamental fields and monopole operators modulo superpotential relations.  
At arbitrary $N$, the superpotential for these $\mathcal{N} =3$ necklace quivers when $n_F=0$ is
\be
W = \sum_{a=1}^d \frac{1}{k_a} \operatorname{Tr}(B_{a+1} A_{a+1} - A_{a}B_{a})^2 \ ,
\ee
where $k_a = q_{a+1}-q_a$.  When $N=1$, this superpotential gives rise to $d-2$ linearly independent relations on the geometric branch of the moduli space
\be
(k_a + k_{a-1}) A_a B_a = k_{a-1} A_{a+1} B_{a+1} + k_{a} A_{a-1} B_{a-1} \ .
\ee
(As pointed out by \cite{Jafferis:2008qz}, these relations are one set of moment map constraints in the hyperk\"ahler quotient discussed previously.)  There also exists a monopole operator $T$ and anti-monopole operator $\tilde T$ that create one and minus one units of flux, respectively, through each gauge group.
These monopole operators satisfy the quantum relation $T \tilde T = 1$ \cite{Jafferis:2009th, Benini:2009qs}.
The $A_a$ and $B_a$ fields have $R$-charge 1/2, guaranteeing that the $R$-charge of the superpotential is two.  We will take the monopole operators to have zero $R$-charge.
Gauge invariance for this $U(1)^d$ gauge theory means that the total $U(1)^d$ charge of a gauge-invariant operator constructed from $A_a$, $B_a$, and monopole fields will vanish.    Let $A_a$ ($B_a$) have charge $+1(-1)$ under gauge group $a$ and charge $-1(+1)$ under gauge group $a-1$.  The monopole operators have gauge charges $\pm k_a$ under the $a$th gauge group. To write a gauge-invariant operator in a compact form, it is convenient to define the operators
\be
C_a^{m q_a + s} \equiv
\begin{cases}
A_a^{mq_a + s} \  & \mbox{if  } mq_a +s > 0 \  \\
B_a^{-m q_a - s} \ & \mbox{if  } mq_a + s < 0 \ 
\end{cases}
\; ; \qquad
U^m \equiv \begin{cases}
T^m \ & \mbox{if  } x > 0 \\
\tilde T^{-m} \ & \mbox{if  } x<0
\end{cases} 
\ .
\ee
The gauge-invariant operators are
\be
\mathcal{O}(m,s,i,j)  = U^m C_1^{m q_1 + s} C_2^{m q_2 + s} \cdots C_d^{m q_d + s} (A_1 B_1)^i (A_2 B_2)^j
\ .
\ee
Note we have used the superpotential relations to eliminate all but two of the $(A_i B_i)$ factors that could potentially appear in a gauge-invariant operator.

Let us consider the set of operators of the form $\mathcal{O}(m,s,0,0)$ 
with $R$-charge less than $r$.  
The number of these operators is equal to the number of lattice points inside the polygon
\be
\mathcal{P}_r = \left\{ (m,s) \in \mathbb{R}^2 : \frac{1}{2} \sum_{a=1}^d | s + q_a m  | < r \right\} \ .
\ee
In the large $r$ limit, the number of these lattice points 
is well approximated by $\operatorname{Area}(\mathcal{P}_r)$. 
The polygon in Section \ref{VOLUME} and this polygon are related via
$\mathcal{P}_{1/2} = \mathcal{P}$.  
Corollary \ref{cor:polygon} established that $\Area(\mathcal{P})$
and $\Vol(Y)$ are proportional.  As $\Area(\mathcal{P}_r) {=} 4 r^2 \Area(\mathcal{P})$, 
we have an additional relation between the number of a certain type of operator and
$\Vol(Y)$.  

The relation between $\Area(\mathcal{P})$ and an operator counting problem reveals an additional relationship between the eigenvalue density $\rho(x)$ and the chiral ring.  
We claim that in the large $r$ limit the number of operators $\mathcal{O}(m,s,0,0)$ of $R$-charge less than $r$ and monopole charge between $m=x r / \mu$ and $m+dm = (x + dx) r / \mu$ is
\be
 \frac{r^2}{\mu^2} \rho(x) dx \ .
\ee
From (\ref{Soln}), the quantity $\rho(x) dx$ in the matrix model corresponds to the area of a constant $x$ strip of $\mathcal{P}_{\mu}$ of height $s_L(x) - s_S(x)$ and width $dx$.  
Roughly speaking, 
the operator counting gives $\rho(x)$ a new interpretation as the number of operators
$\mathcal{O}(x,s,0,0)$
of $R$-charge less than $ \mu$ and monopole charge bounded between $x$ and $x+dx$.
As $\mu$ is of order one in the matrix model, we should be more careful and consider first
$\mathcal{P}_r$ for some large $r$, giving rise to the factors of $r/\mu$ in the claim.

There is a subtle relation between $\delta y_a$ and the chiral ring where $A_a$ or $B_a$ is set to zero.  
In the large $r$ limit, we claim that 
\be
 \frac{r^2}{\mu^2} \left(\delta y_a(x) + \frac 12 \right) \rho(x) dx
 \label{dycountingeq}
\ee
counts the number of operators of the form $\mathcal{O}(m,s,0,0)$ with $R$-charge less than $r$ and with monopole charge between $m = x r / \mu$ and $m+dm = (x + dx) r / \mu$ and with $B_a=0$.
Flipping the sign of $\delta y_a(x)$ yields an equivalent expression for operators with $A_a=0$.  Note that
\begin{eqnarray}
\rho(x) \left(\delta y_a(x) + \frac 12 \right) &=& 
\begin{cases}
0 &  \ q_a x < -s_L(x)  \,, \\
s_L(x) + q_a x &  \ -s_L(x) < q_a x < -s_S(x)\,,  \\
 s_L(x) - s_S(x)  &  \ -s_S(x) < q_a x \,.
\end{cases}
\end{eqnarray}
For simplicity, let us first assume that $\mu$ is large and count the operators with $R$-charge
less than $\mu$.  We get the correct counting when $q_a x < -s_L(x)$ and when $-s_S(x) < q_a x$ because 
the region $q_a x + s_S(x) > 0$ corresponds to a portion of the polygon where the
$\mathcal{O}(x,s,0,0)$ contain no $B_a$ while the region $q_a x + s_L(x) < 0$ corresponds to operators that contain no $A_a$.  In the central region, the operators that contain no $B_a$ satisfy the constraint $s_L(x) > s > -q_a x$.  Because there is one operator per lattice point, the number of operators that contain no $B_a$ is proportional to the difference
$s_L(x) - (-q_a x)$.  
Similarly $\rho(x) (-\delta y_a(x) + 1/2)$ will count the number of operators with no $A_a$.
Again, since $\mu$ is of order one in the matrix model, we should rescale our results for a polygon $\mathcal{P}_r$ in the large $r$ limit, yielding the extra factors of $r/\mu$ in the claim  (\ref{dycountingeq}) for $\delta y_a$.

For a general supersymmetric gauge theory, we will not be able to make a clean separation between all operators in the chiral ring and operators of this special form $\mathcal{O}(m,s,0,0)$.  It is thus useful to reformulate these statements about $\rho(x)$ and $\delta y_a$ in terms of all chiral operators. 
Let us introduce the function $\psi(r,m)$ which counts the number of operators with $R$-charge less than $r$ and monopole charge less than $m$.  We claim that
\be
\left. \frac{\partial^3 \psi}{\partial r^2  \partial m} \right|_{m=rx/ \mu} = \frac{r}{\mu} \rho(x) \ .
\label{psirhorel}
\ee
Note that in the large $r$ limit, $(\partial^2 \psi / \partial r \partial m) dr \, dm$ can be interpreted as the number of operators of $R$-charge between $r$ and $r+dr$ and monopole charge between $m$ and $m+dm$.
Given an operator $\mathcal{O}(m,s,0,0)$ of $R$-charge
$\frac{1}{2} \sum_a |s + q_a m| = r_0 < r$, we can form an operator of $R$-charge equal to $r$ by multiplying $\mathcal{O}(m,s,0,0)$ by a factor of $(A_1B_1)^j (A_2 B_2)^{r-r_0-j}$.  There are precisely $r-r_0+1$ ways of forming such a factor (assuming $r-r_0$ is an integer).  This multiplicity $r-r_0+1$ associated to a lattice point $(m,s)$ of the polygon can be interpreted as the total number of operators of fixed $m$ and $s$ with $R$-charge equal to $r$.  The difference in the number of operators with $R$-charge $r+1$ and $r$ integrated over a strip at constant $m$ now has the dual interpretation as $r \rho(\mu m / r) / \mu$ or $\partial^3 \psi / \partial r^2  \partial m$ (in the large $r$ limit).

As for $\delta y_a(x)$, let us introduce $\psi_{A_a}(r,m)$ and $\psi_{B_a}(r,m)$ as the number of operators with $R$-charge less than $r$ and monopole charge less than $m$, with $A_a=0$ and $B_a=0$ respectively. We claim that 
\begin{eqnarray}
\label{psiAdyrel}
\left. \frac{\partial^2 \psi_{A_a} }{\partial r \partial m}\right|_{m=rx/ \mu} &=& \frac{r}{\mu} \rho(x) \left(-\delta y_a(x) + \frac 12 \right)\ , \\
\label{psiBdyrel}
\left. \frac{\partial^2 \psi_{B_a} }{\partial r \partial m}\right|_{m=rx/ \mu}  &=& \frac{r}{\mu} \rho(x) \left(\delta y_a(x) +\frac 12\right) \ .
\end{eqnarray}
Note that $(\partial^2 \psi_{B_a} / \partial r \partial m) dr \, dm$ can be interpreted as the number of operators of $R$-charge between $r$ and $r+dr$, monopole charge between $m$ and $m+dm$, and no $B_a$ operators, in the large $r$ limit.
Given an operator $\mathcal{O}(m,s,0,0)$ of $R$-charge $r_0<r$ that does not involve the field $B_a$, we can form a unique 
operator of $R$-charge equal to $r$ by multiplying $\mathcal{O}(m,s,0,0)$ by $(A_b B_b)^{r-r_0}$ where $b \neq a$ (provided $r-r_0$ is an integer).  Thus, eq.\ (\ref{dycountingeq}) is counting the number of operators in the chiral ring with $R$-charge between $r$ and $r+dr$, magnetic charge between $m=r x/\mu$ and $m+dm = r(x+dx)/\mu$, and no $B_a$ operators.  There is a parallel argument for $\psi_{A_a}$.

\subsection{Operator Counting and Volumes}

Given the close relation between $\rho(x)$ and $\Vol(Y)$, it is not surprising that there is also a close connection between $\rho(x)$ and numbers of operators in the chiral ring.  
Motivated by Weyl's Law for eigenfunctions of a Laplacian on a curved manifold, the authors of \cite{Bergman:2001qi} noticed that the number of holomorphic functions on a certain class of Calabi-Yau cones could be related to the volume of the Sasaki-Einstein manifold base.  
Their result holds for a Calabi-Yau cone which is a $\mathbb{C}^*$-fibration over a variety in weighted projective space with a K\"ahler-Einstein metric.  This special case was later generalized to any complex K\"ahler cone by Martelli, Sparks, and Yau \cite{Martelli:2006yb}.
Note that the set of holomorphic functions on the cone is precisely the chiral ring of the $N=1$ gauge theory and that the ring has a natural grading from the $R$-charge.  
The relation between the number of holomorphic functions and the $\Vol(Y)$ is
\be
\Vol(Y) = \frac{\pi^n n}{2^{n-1}} \lim_{r\to \infty} \frac{1}{r^4}\lim_{m \to \infty} \psi(r,m) \ ,
\label{MSYresult}
\ee
where $\dim(Y) = 2n-1$. 
In fact the result (\ref{MSYresult}) helps to explain not only the relation (\ref{psirhorel}) between $\rho(x)$ and $\psi(r,m)$ but also the relations (\ref{psiAdyrel}) and (\ref{psiBdyrel}) between $\delta y_a(x)$ and restrictions of the chiral ring to rings where $A_a=0$ or $B_a=0$. 
The subspace $A_a=0$ of the Calabi-Yau cone is also a complex cone with K\"ahler structure.  The level surfaces of this complex cone will be Sasaki manifolds which satisfy (\ref{MSYresult}) with $n=3$.  
The identifications (\ref{psirhorel}), (\ref{psiAdyrel}), and (\ref{psiBdyrel}) along with (\ref{MSYresult}) imply that
\begin{eqnarray}
\label{rhoVolrel}
\int \rho(x) dx &=& \frac{24 \mu^2}{\pi^4} \Vol(Y) \ , \\
\int \rho(x) \left(-\delta y_a(x) + \frac 12 \right) dx &=&  \frac{4 \mu^2}{\pi^3} \Vol(Y_{A_a}) \ , \\
\int \rho(x) \left(\delta y_a(x) + \frac 12  \right) dx &=&  \frac{4 \mu^2}{\pi^3} \Vol(Y_{B_a}) \ .
\end{eqnarray}
The relation (\ref{rhoVolrel}) we derived already, but the second two relations are seemingly new.  
From the free energy functional (\ref{necklaceF}), it is clear that $\delta y_a(x)$ is an odd function of $x$, and thus for these necklace quivers the integral $\int \rho(x) \delta y_a(x) dx$ vanishes trivially.  The remaining integral yields the result
\be
\Vol(Y_{A_a}) = \Vol(Y_{B_a}) = \frac{3}{\pi} \Vol(Y) \ ,
\ee  
a result Yee found by other means \cite{Yee:2006ba}.

A field theoretic interpretation of these five dimensional cycles was provided by \cite{Gubser:1998fp,Berenstein:2002ke}.  In an $AdS_4 \times Y$ solution of eleven-dimensional supergravity, an M5-brane wrapping such a cycle in $Y$ looks like a point particle in $AdS_4$ with a mass proportional to the volume of the cycle times the tension of the five-brane.  The AdS/CFT dictionary provides a relationship between the mass of the particle and its conformal dimension.  As the wrapped five-brane is supersymmetric, the conformal dimension can be related to the $R$-charge of the corresponding operator that creates the state.  If in the geometry the five-cycle corresponds to setting $A_a=0$, the corresponding baryonic-like operator should involve an anti-symmetric product of $N$ copies $A_a$.  For our purposes, the essential point is a relation between $\Vol(Y_{A_a})$ and the $R$-charge of $A_a$:
\be
R(A_a) = \frac{\pi}{6} \frac{\Vol(Y_{A_a})}{\Vol(Y)} = \frac{1}{2} \ .
\ee
This point suggests a way of generalizing  (\ref{psiAdyrel}) and (\ref{psiBdyrel}) to an arbitrary quiver gauge theory.  We should replace the $1/2$ with the $R$-charge of the corresponding bifundamental field $X_{ab}$ with charge $+1$ under gauge group $b$ and charge $-1$ under gauge group $a$:
\be
\left. \frac{\partial^2 \psi_{X_{ab}}}{\partial r \partial m} \right|_{m = rx/\mu}
=
\frac{r}{\mu} \rho(x) \left[y_b(x) - y_a(x) + R(X_{ab}) \right] \ .
\ee

\section{The Matrix Model for $(p, q)$-Branes}

As discussed in section~\ref{BRANE}, the free energy for a system of $d$ $(p_a, q_a)$ five-branes can be computed on the gravity side by combining the M-theory prediction \eqref{MtheoryExpectation} with the eq.~\eqref{VolYVolSphere}, where $n=2$ and $(p_a, q_a) = \beta_a^T$.  More explicitly, the polygon ${\cal P}$ takes the form
 \es{Polytopepq}{
  {\cal P} = \left\{(x,s) \in \R^2: \sum_{a=1}^d \abs{p_a s + q_a x} \leq 1 \right\} \,,
 }
and the volume of ${\cal P}$ is related to $F$ through $F =2 \pi N^{3/2} / 3 \sqrt{\Vol({\cal P})}$.  In the previous section we were able to reproduce this formula from the field theory side in the cases where $p_a = 0$ or $p_a = 1$ for each $a$, as these were the cases where a simple Lagrangian description of the field theory led \cite{Kapustin:2009kz} to an expression for $F$ in terms of a matrix integral. For $p_a > 1$, as mentioned in the introduction, a Lagrangian description would involve the coupling to the $T(U(N))$ theory described in \cite{Gaiotto:2008ak}.  We will now use our large $N$ intuition to figure out the matrix model at finite $N$.

As a first step towards obtaining such a matrix model, we note that at large $N$ one can obtain the correct value for $F$ by extremizing the free energy functional:
\be
F[\rho, \delta y_a] = \pi N^{3/2} \int \rho(x) dx \left[
2 x \sum_{a=1}^d \frac{q_a}{p_a} \delta y_a(x) + \rho(x) \sum_{a=1}^d p_a 
f(\delta y_a(x) / p_a) \right] \ .
\label{necklaceFgenp}
\ee
The saddlepoint is a generalization of the earlier eigenvalue distribution (\ref{Soln})
for the $(1, q_a)$-branes: 
\es{Solnpq}{
  \rho(x) &= s_L(x) - s_S(x) \,, \qquad 
   \nu(x) = - \frac{1}{2} \left[ s_L(x) + s_S(x)\right]  \,, \\
   \delta y_a(x) &= \frac{1}{2} \frac{\abs{p_a s_L(x) + q_a x }  
     -\abs{p_a s_S(x) + q_a x} }{s_L(x) - s_S(x)}  \,,
 }
 where $s_L(x) > s_S(x)$ are the two solutions of 
 $
\frac{1}{2} \sum_{a=1}^d \abs{p_a s + q_a x} = \mu 
$.

The equation (\ref{necklaceFgenp}) seems ill behaved in the limit $p_a \to 0$, but in fact it is not.  
What happens is that $\delta y_a(x)$ will saturate very quickly to $\pm p_a/2$ as we move away from $x=0$.  Thus the function $f(\delta y_a/ p_a)$ will vanish, and 
$2 x q_a  \delta y_a / p_a$ can be replaced with $q_a |x|$.  Before, we identified the number of flavors with the sum 
$
n_F = \sum_{p_a = 0} q_a
$, and in this way we see how a term of the form $n_F |x|$ will appear in the free energy
(\ref{necklaceFgenp}).

To give the reader a better sense of how $SL(2, \mathbb{Z})$ acts on the polygon, we show $\mathcal{P}$ for ABJM theory and some of its $SL(2, \mathbb{Z})$ transforms in figure \ref{fig:poly}.
While all three polygons have the same volume (as they must given Corollary \ref{cor:n2ordered}),
the eigenvalue distributions can look quite different.  In figure \ref{fig:poly}a, 
$\rho(x)$ is a constant function for $-1/k< x < 1/k$.  In figure \ref{fig:poly}b, $\rho(x)$ has two piecewise linear regions for $-1 < x < 1$.  Finally, in figure \ref{fig:poly}c, $\rho(x)$ has one constant region and two linear regions in the interval $-1 + 1/k < x < 1 - 1/ k$.

\begin{figure}
a) 
\includegraphics[width=4.5cm]{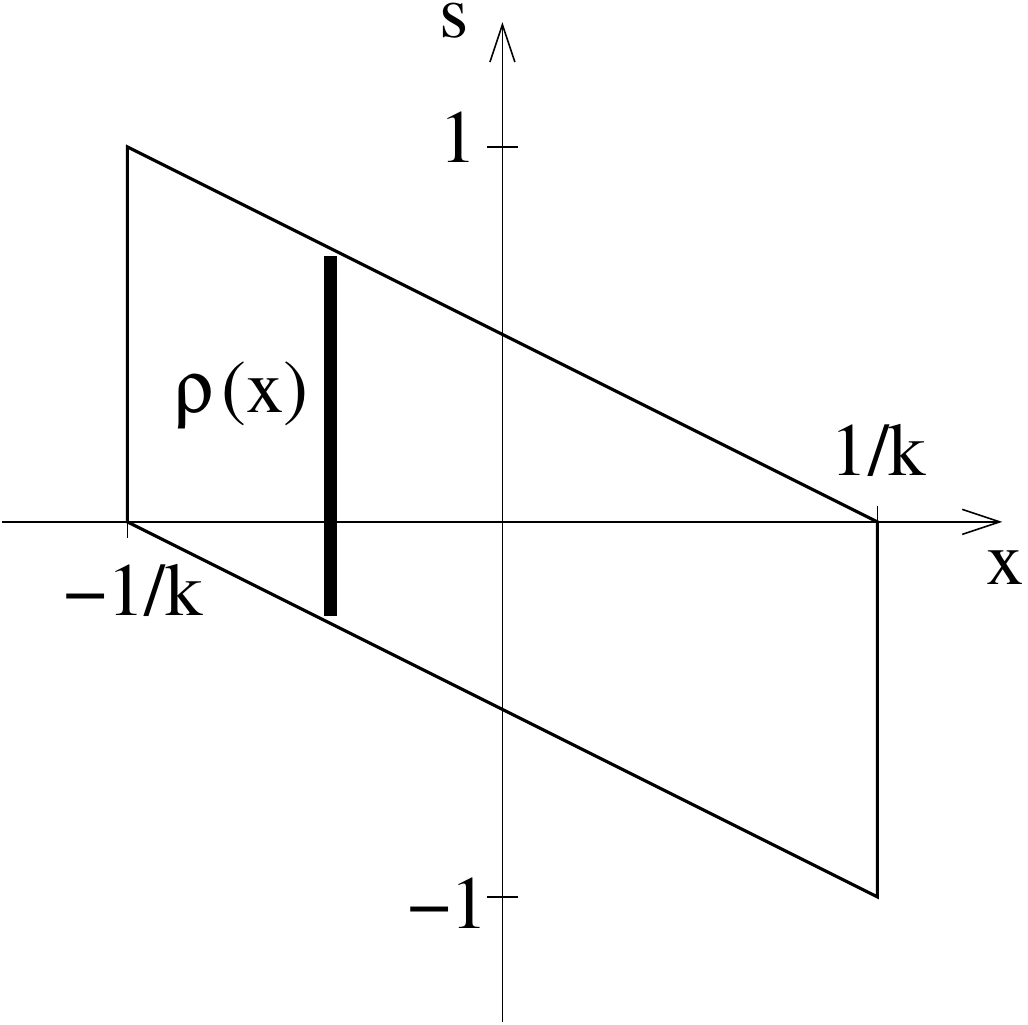}
b)
\includegraphics[width=4.5cm]{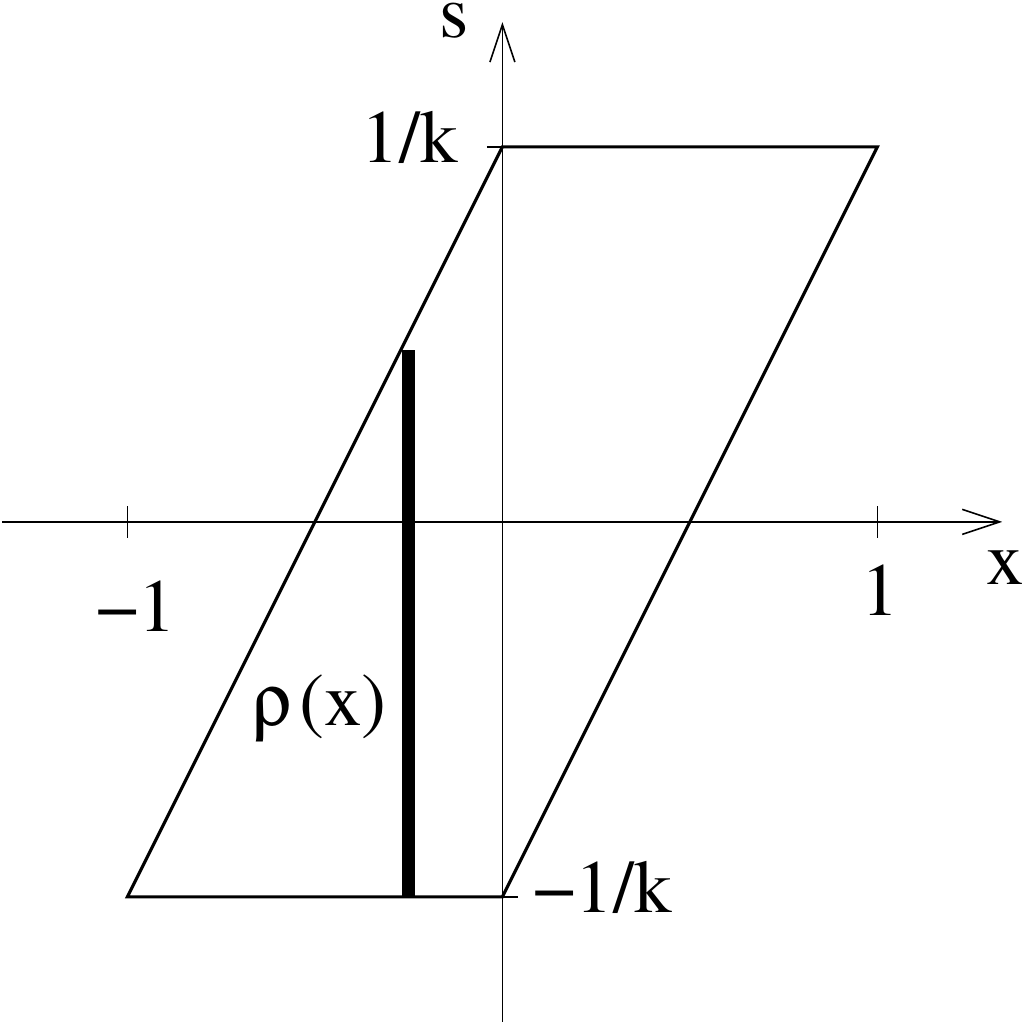}
c)
\includegraphics[width=4.5cm]{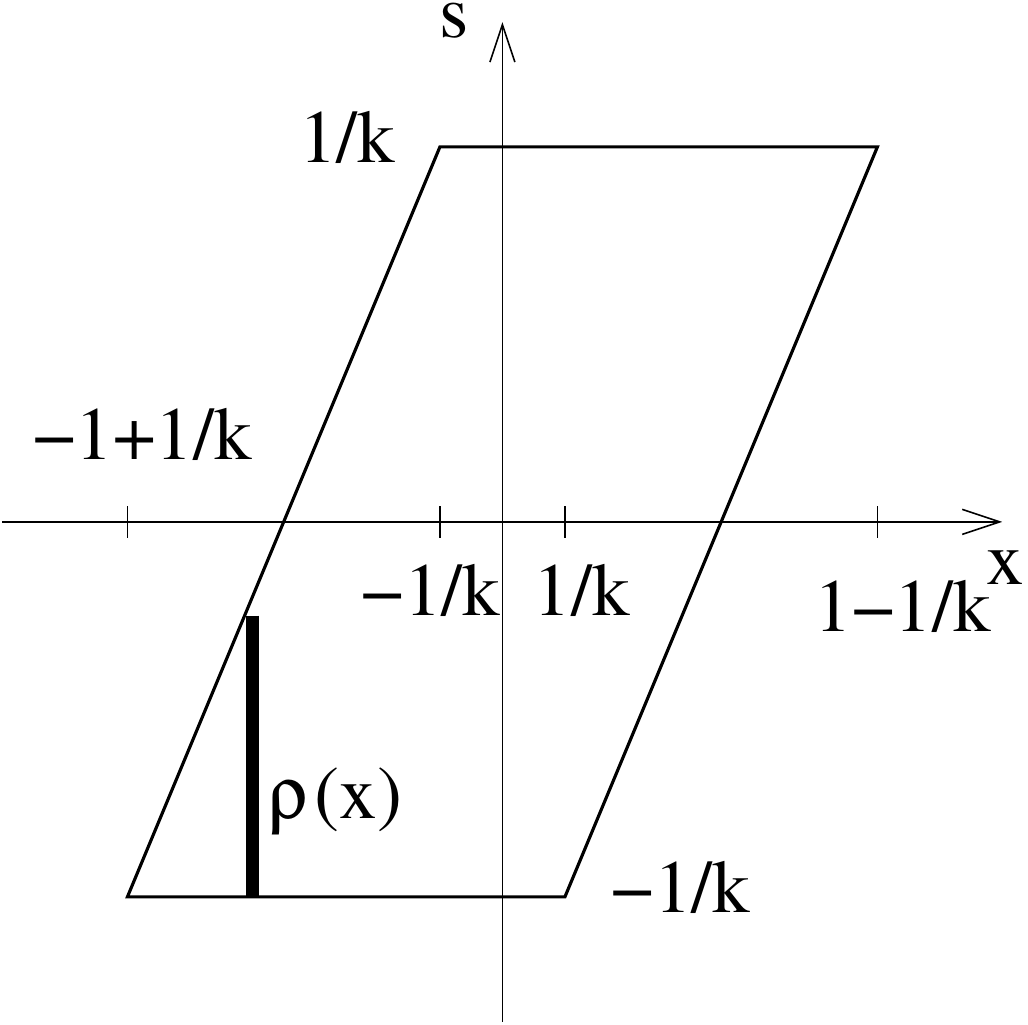}
\caption{
a) The polygon for ABJM theory which can be built from a $(1,0)$ and $(1,k)$ brane;
b) the S-dual configuration involving a $(0,1)$ and $(-k,1)$ brane;
c) an $SL(2,\mathbb{Z})$ transform to a $(1,1)$ and $(1-k, 1)$ brane.
}
\label{fig:poly}
\end{figure}

\subsection{$(p,q)$-branes at finite $N$}

One can go a little further and conjecture a finite $N$ analog of the matrix integral \eqref{NecklaceModel} that reduces to \eqref{necklaceFgenp} in the large $N$ limit.  If we move the $\sinh$ and $N!$ factors from \eqref{Lv} to \eqref{Lm}, then
the contribution from a $(1,q_a)$-brane to the partition function is given by
\es{oneqbrane}{
&\frac{1}{N!} \exp \left[\pi i q_a \left(\sum_i \lambda_{a-1,i}^2 - \sum_i \lambda_{a,i}^2 \right) \right] \\ 
&\qquad \qquad \times \frac{\prod_{i<j} 2 \sinh \pi (\lambda_{a-1,i}-\lambda_{a-1,j}) \prod_{i<j} 2 \sinh \pi (\lambda_{a,i}-\lambda_{a,j})}{\prod_{i,j} 2 \cosh \pi(\lambda_{a-1,i} - \lambda_{a,j})} \,.
}
We argue that the generalization to $(p_a,q_a)$ branes is given by
\be \label{eq:partitionpq}
\begin{split}
L_{(p_a, q_a)}(\lambda_{a-1}, \lambda_a) \equiv & \frac{1}{N!} |p_a|^{-N} \exp \left[\pi i \frac{q_a}{p_a} \left( \sum_i \lambda_{a-1,i}^2 - \sum_i \lambda_{a,i}^2 \right) \right] \\
& \times \frac{\prod_{i<j} 2 \sinh \frac{\pi}{p_a} (\lambda_{a-1,i}-\lambda_{a-1,j})
\prod_{i<j} 2 \sinh \frac{\pi}{p_a} (\lambda_{a,i}-\lambda_{a,j})}{\prod_{i,j} 2 \cosh \frac{\pi}{p_a}(\lambda_{a-1,i} - \lambda_{a,j})} \,,
\end{split}
\ee
where this formula is correct only when $(p_a, q_a)$ are relatively prime.  An $(np_a, nq_a)$ five-brane, with $(p_a, q_a)$ relatively prime, should be thought of as $n$ $(p_a, q_a)$ five-branes.
Eq.~\eqref{eq:partitionpq} is based largely on the structure of (\ref{necklaceFgenp}).  The term involving the hyperbolic cosine in (\ref{oneqbrane}) gives rise to the second term in (\ref{necklaceFgenp}) in the case $p_a=1$.  As the imaginary parts of the eigenvalues appear with a factor of $p_a$ in (\ref{necklaceFgenp}), 
at finite $N$, we should divide the eigenvalue 
differences $\lambda_{a-1,i} - \lambda_{a,j}$ by $p_a$.  
The exponential term in (\ref{oneqbrane}) gives rise to the first term in 
(\ref{necklaceFgenp}) in the case $p_a=1$. As a result at finite $N$ 
we should replace the coefficient $q_a$ with $q_a/p_a$ in (\ref{oneqbrane}).  
This replacement suggests that the field theory on $(p_a,q_a)$-branes in some sense can be thought of as having Chern-Simons couplings $\frac{1}{4 \pi} (q_{a+1}/p_{a+1} - q_{a}/p_{a})  \int \tr A_a \wedge dA_a$.  (Such an observation was also made in \cite{Kitao:1998mf, Gaiotto:2008ak}.)
Recall that the factor $\exp[\pi i (q_{a+1} - q_{a}) \sum_i \lambda_{a, i}^2 ]$ in the original matrix model comes from the classical contribution of the Chern-Simons term $\frac{1}{4 \pi} (q_{a+1} - q_{a})  \int \tr A_a \wedge dA_a$ and its supersymmetric completion.

The remaining factors of $p_a$ are required for (\ref{eq:partitionpq}) to be invariant under S-duality as we will now see.  As an added bonus, by studying the action of S-duality, we will be able to deduce the partition function for the $T(U(N))$ theory of \cite{Gaiotto:2008ak}.

Our $(p,q)$-brane construction exists in type IIB string theory which is well known to be invariant under the action of $SL(2, \mathbb{Z})$.  One of the generators of $SL(2, \mathbb{Z})$ is S-duality which we define to map a $(p,q)$-brane to a $(-q,p)$-brane.  
The work of Gaiotto and Witten \cite{Gaiotto:2008ak} suggests that we also should be able to realize S-duality locally, on one $(p,q)$-brane at a time.
The fact that S-duality squares to minus one suggests that we may be able to realize 
it as a Fourier transform acting on 
(\ref{eq:partitionpq}).  

For simplicity, we will restrict to the case where the ranks of the gauge groups are equal to $N$. 
Before introducing the Fourier transform, we make use of the identity
\begin{equation} \label{DeterminantEqual}
\begin{split}
& \frac{\prod_{i<j} \sinh(x_j - x_i) \prod_{i<j} \sinh(y_j-y_i)}{\prod_{i,j} \cosh(x_i-y_j)} \\ = & \det \left(
\begin{array}{cccc}
\sech(x_1-y_1) & \sech(x_2-y_1) &
\dots & \sech(x_n-y_1) \\
\sech(x_1-y_2) & \sech(x_2-y_2) &
\dots & \sech(x_n-y_2) \\
\vdots & \vdots & dots & \dots \\
\sech(x_1-y_N) & \sech(x_2-y_N) &
\dots & \sech(x_N-y_N)
\end{array}
\right)
\end{split}
\end{equation}
previously considered in \cite{Kapustin:2010xq}.
Given this identity, we can write (\ref{eq:partitionpq}) as a sum over permutations
\begin{equation} \label{pqFactor}
L_{(p, q)}(\lambda, \sigma ) = \frac{1}{N!} \sum_{\rho \in S_N} (-1)^\rho \frac{1}{|p|} 
\prod_j 
\exp\left[ \pi i \frac{q}{p} (\lambda_{j}^2 - \sigma_{\rho(j)}^2 ) \right] \frac{1}{2 \cosh \frac{\pi}{p}(\lambda_{j}-\sigma_{\rho(j)} )}  \ .
\end{equation}

We claim that a local S-duality is implemented by the following Fourier transform:
\be
L_{(-q,p)}(\mu, \nu) =
\int  
e^{2 \pi i \mu \cdot \lambda} L_{(p,q)}(\lambda, \sigma)
e^{-2 \pi i  \nu \cdot \sigma}
d^N \lambda \, d^N \sigma  \ .
\label{Sdualityresult}
\ee 
To demonstrate this claim, we isolate the integrals over $\lambda_j$ and $\sigma_{\rho(j)}$:
 \es{IDef}{
I = \frac{1}{|p|}
  \int d \lambda_j d \sigma_{\rho(j)}  
e^{ \pi i \left( \frac{q}{p}( \lambda_j^2 -  \sigma_{\rho(j)}^2 )
+2 ( \mu_j  \lambda_j - \nu_{\rho(j)}  \sigma_{\rho(j)} )
 \right)}
\frac{1}{\cosh \frac{\pi}{p}( \lambda_j - \sigma_{\rho(j)})} \ .
 }
With the change of variables $x_\pm =  \lambda_j \pm  \sigma_{\rho(j)}$, this integral is straightforward to perform:
\begin{eqnarray}
I &=&
\frac{1}{2|p|}
\int d x_+ \, dx_- \, e^{\pi i \left( \frac{q}{p} x_+ x_- + x_+(\mu_j - \nu_{\rho(j)})
+ x_-(\mu_j + \nu_{\rho(j)}) \right)}
\sech \frac{\pi x_-}{p} 
\nonumber 
\\
&=&
\frac{1}{ |p|}
\int dx_- \, \delta\left(
\frac{q}{p} x_- + \mu_j - \nu_{\rho(j)} \right) e^{\pi i x_-(\mu_j + \nu_{\rho(j)})} 
\sech \frac{\pi x_-}{p} 
\\
&=&
\frac{1}{|q|}
\exp \left[ - \pi i \frac{p}{q} ( \mu_j^2 - \nu_{\rho(j)}^2) \right]
\frac{1}{\cosh \frac{\pi}{q} ( \mu_j - \nu_{\rho(j)} ) } \ . \nonumber
\end{eqnarray}
Taking the product over the eigenvalues and averaging over permutations yields (\ref{Sdualityresult}).

This local S-duality composes in a nice way.  Consider applying similar Fourier transforms to neighboring $(p, q)$-branes:
\begin{eqnarray}
\lefteqn{\int L_{(-q,p)}(\lambda, \mu) L_{(-q',p')}(\mu, \nu) d^N \mu 
 } \nonumber \\
 &=&
\int e^{2 \pi i \lambda \cdot \tilde \lambda} L_{(p,q)}(\tilde \lambda, \mu_1)
e^{2 \pi i \mu \cdot (\mu_2 - \mu_1)} L_{(p',q')} ( \mu_2, \tilde \nu) 
e^{-2 \pi i \tilde \nu \cdot \nu} d^N \tilde \lambda \,  d^N \mu_1 \, d^N \mu  \, d^N \mu_2 \, d^N \tilde \nu \nonumber \\
&=&
\int 
e^{2 \pi i \lambda \cdot \tilde \lambda} L_{(p,q)}(\tilde \lambda, \tilde \mu)
L_{(p',q')} ( \tilde \mu, \tilde \nu) 
e^{-2 \pi i \tilde \nu \cdot \nu} d^N \tilde \lambda \, d^N \tilde \nu \, d^N \tilde \mu
\ . 
\end{eqnarray}
Thus if we apply a local S-duality to each $(p_a, q_a)$-brane in the necklace, the factors of $e^{2 \pi i \mu \cdot \lambda}$ cancel out and the resulting partition function is invariant under a global action of S-duality.

We would like to give a better interpretation of this group action.
Consider acting on a single $(p,q)$-brane with this local S-duality:
 \es{SingleAction}{
\int L_{(-q,p)}(\lambda, \mu) L_{(p',q')}(\mu, \nu) d^N \mu =
\int e^{2 \pi i \lambda \cdot \tilde \lambda} L_{(p,q)}(\tilde \lambda, \tilde \mu)
e^{-2 \pi i \mu \cdot \tilde \mu } L_{(p',q')} ( \mu, \nu)  d^N \tilde \lambda \,  d^N \tilde \mu \, d^N \mu \ .
 }
One way of interpreting the $e^{-2 \pi i \mu \cdot \tilde \mu}$ is to posit that some new object has been inserted between the $(p,q)$-brane and the $(p',q')$-brane that implements a local S-duality.  This object contributes to the partition function an amount
 \es{LSDef}{
L_S(\mu, \tilde \mu) \equiv e^{-2 \pi i \mu \cdot \tilde \mu} \ .
 }
Similarly, right before the $(p, q)$ brane we introduced another object that undoes the local S-duality:
 \es{SInverse}{
  L_{S^{-1}}(\lambda, \tilde \lambda) \equiv e^{2 \pi i \lambda \cdot \tilde \lambda} \,.
 }

Let us see how the $T(U(N))$ theory arises.  So far, we have been thinking of the $(p,q)$-branes as the building blocks out of which we construct the partition function.  Alternately, we can decompose the partition function into the contributions from the D3-brane segments and associated $U(N)$ gauge groups.  
From the D3-brane point of view,  the object $L_S(\mu, \tilde \mu)$ implementing S-duality splits a D3-brane segment 
into two regions, each characterized by a $U(N)$ gauge theory, one with eigenvalues $\mu$ and one with eigenvalues $\tilde \mu$.  
The $L_{(p,q)}$ and $L_{(p',q')}$ factors do not have enough factors of
hyperbolic sine to describe two $U(N)$ gauge theories.  These factors are simple to recover if we say the new object has a partition function
 \es{LSTildeDef}{
\tilde L_S(\mu, \tilde \mu) = N!\,e^{-2 \pi i \mu \cdot \tilde \mu} \prod_{i < j} \frac{1} {4 \sinh \pi (\mu_i - \mu_j)  \sinh \pi ( \tilde \mu_i - \tilde \mu_j)} \ .
 }
This object has a natural interpretation as the partition function of the $T(U(N))$ theory.\footnote{Actually, the $T(U(N))$ partition function needs to
be antisymmetrized with respect to permuting $\mu$.  But since
$L_{(p,q)}(\mu, \lambda)$ is already antisymmetric in $\mu$, we can
get away with not antisymmetrizing $L_S$.}
This partition function has been found independently by \cite{Benvenuti:2011ga}.

At this point, it is clear that we should be able to realize any element of $SL(2, \mathbb{Z})$ acting locally on our necklace theories.  The group $SL(2, \Z)$ has two generators: $S$, which we discussed above, and $T$.  We define $T$ to send $q \to q + p$ and leave $p$ invariant.  If we think of $(p, q)$ as a two-component column vector on which $SL(2, \Z)$ acts in the fundamental representation, then $S$ and $T$ are the two-by-two matrices
 \es{STMatrices}{
  S = \begin{pmatrix} 0 & -1 \\
   1 & 0
   \end{pmatrix} \,, \qquad
   T = \begin{pmatrix} 1 & 0 \\
    1 & 1
    \end{pmatrix} \,.
 }
They satisfy $S^2 = -1$ and $(ST)^3 = 1$.  To find the action of $T$ on the matrix model, one can see from \eqref{pqFactor} that  
 \es{TDuality}{
L_{(p,q+p)}(\lambda, \sigma) = e^{\pi i \lambda \cdot \lambda - i \theta} L_{(p,q)}(\lambda, \sigma) 
e^{-\pi i \sigma \cdot \sigma + i \theta} \ ,
 }
where $\theta$ is a phase to be determined.
Similarly to $L_S$, one can therefore define
 \es{LTDefs}{
L_T(\sigma) \equiv e^{-\pi i \sigma \cdot \sigma + i \theta} \,, \qquad
 L_{T^{-1}} (\lambda) \equiv e^{\pi i \lambda \cdot \lambda - i \theta} \,,
 }
so one can describe the contribution of a $(p, q+p)$ five-brane to the matrix model integrand as $L_{T^{-1}} (\lambda) L_{(p, q)}(\lambda, \sigma) L_T(\sigma)$, where the factor of $L_T(\sigma)$ corresponds to the local action of a $T$-transformation and $L_{T^{-1}}(\lambda)$ corresponds to the local action of $T^{-1}$.

We have defined the effect of the $S$ and $T$ generators on the matrix model so that $S$ requires us to have two distinct sets of eigenvalues $\mu$ and $\tilde \mu$ in the two regions separated by an ``$S$-boundary'', while $T$ doesn't, acting just by multiplication of the eigenvalues $\sigma$ in the region containing the ``$T$-boundary'' by $L_T(\sigma)$.  We could have said, however, that a $T$-boundary also requires two distinct sets of eigenvalues $\sigma$ and $\tilde \sigma$ around it, in which case we should have described its contribution as $L_T(\sigma) \delta(\sigma - \tilde \sigma)$.  By the same logic, it follows that the identity element should be also viewed as a Dirac delta function $L_1(\lambda, \sigma) = \delta (\lambda-\sigma)$.  The operator equal to minus the identity is  $L_{-1}(\lambda, \sigma) = \delta (\lambda+\sigma)$.  
We can verify by explicit computation that
 \es{SSInverse}{
L_{SS^{-1}}(\mu, \lambda) = \int e^{2 \pi i \sigma \cdot(\lambda- \mu)} d^N \sigma = 
\delta( \mu - \lambda) = L_1(\mu, \lambda) \,,
 }
and similarly that  $L_{TT^{-1}} = L_1$, $L_{S^2} = L_{-1}$ and $L_{(ST)^3} = L_1$
provided $3 \theta = \pi N /4$.  

We should mention that the expressions \eqref{LSDef}, \eqref{SInverse}, and \eqref{LTDefs} can also be justified by analyzing the case $N=1$, where the action of $SL(2, \Z)$ on Chern-Simons theories was described in \cite{Witten:2003ya, Gaiotto:2008ak}.   Indeed, as explained in \cite{Witten:2003ya}, the $T$ generator just shifts the CS level by one unit, so the action changes by $-\frac{1}{4 \pi}  \int A \wedge dA$ plus its supersymmetric completion.  The classical contribution to the partition function from the scalar $\lambda$ in the ${\cal N} = 2$ vector multiplet then gives $L_T(\lambda)$.  The action of $S$ on a CS theory with some gauge field $\tilde A$ consists of introducing another dynamical gauge field $A$ that couples to the topological current $*d\tilde A$.  This coupling takes the form of an off-diagonal Chern-Simons term $-\frac{1}{2 \pi} \int A \wedge d \tilde A$ plus its supersymmetric completion.  If $\mu$ and $\tilde \mu$ are the scalars in the corresponding vector multiplets, the classical contribution to the partition function from this off-diagonal Chern-Simons term is precisely given by \eqref{LSDef}.

\section{Discussion}

Our main results are additional evidence presented in Section \ref{BRANE} for the $F$-theorem conjecture,
the relations (\ref{resultone}) and (\ref{resulttwo}) between numbers of chiral operators and eigenvalue distributions,  and a conjectured form (\ref{eq:partitionpq}) of the matrix model corresponding to a $(p_a, q_a)$-brane construction
in type IIB string theory.  
Each of these results requires some brief discussion.

We would like to investigate further our proposed matrix model (\ref{eq:partitionpq}).  In particular, it would be interesting to see how 
(\ref{eq:partitionpq}) transforms under Seiberg duality. 
 One statement of the $s$-rule \cite{Hanany:1996ie} 
is that a theory breaks supersymmetry for which there exists a Seiberg duality that produces a gauge group with a negative rank.  Given the matrix model's status as a type of supersymmetric index, one expects that the partition function should vanish for theories that violate the $s$-rule \cite{Kapustin:2010mh}.

Regarding the $F$-theorem, we have not constructed any explicit RG flows, either on the gauge theory side or, via the AdS/CFT correspondence, on the gravity side.  Instead, we have posited the existence of reasonable seeming flows, and we have examined $F$ at the IR and UV fixed points.  For example, by adding a mass to fundamental flavors, one should be able to flow from a theory with a $(0,q)$-brane in the UV to one without it in the IR.  The corresponding volume of the tri-Sasaki Einstein manifold will increase, leading to a decrease in $F$.  Similarly, we can consider an RG flow where a $(p,q)$-brane forms a bound state with a $(p',q')$-brane.  Under such a flow, $F$ will also decrease.  
Given the result of refs.\ \cite{Myers:2010tj,Myers:2010xs} described in the introduction, it seems likely that any gravity dual of an RG flow will obey the $F$-theorem.  One way of interpreting our results, given that our flows also obey the $F$-theorem, is that it may be possible to realize these flows as solutions of eleven-dimensional supergravity. 

Given that the operator counting relations (\ref{resultone}) and (\ref{resulttwo}) can be defined for essentially any KWY matrix model and corresponding superconformal field theory, one wonders if they hold more generally.  In a sequel to this paper \cite{Gulotta:2011aa}, 
we investigate the large $N$ limit
of the KWY matrix models for a number of other superconformal field theories, and we find that indeed these relations are always satisfied.  We 
look at necklace quivers with additional adjoint and fundamental fields.  We look at a couple of non-necklace quivers, for example a $\mathbb{Z}_2 \times \mathbb{Z}_2$ orbifold theory.
In order to look at theories with ${\mathcal N}=2$ supersymmetry, we have to generalize the KWY matrix model to allow for arbitrary $R$-charges \cite{Jafferis:2010un}.  
It turns out that in the theories we study, relations (\ref{resultone}) and (\ref{resulttwo}) are valid not just for the correct $R$-charges but for any $R$-charges compatible with the marginality of the superpotential.

One constraint in these investigations is that for chiral theories, the KWY matrix model does not seem to have a large $N$ limit that is compatible with a dual eleven-dimensional supergravity description \cite{Jafferis:2011zi}.  It will be interesting to see if the 
relations (\ref{resultone}) and (\ref{resulttwo}) can give any insight into how the matrix model might be modified to allow for such a limit.

\section*{Acknowledgments}
We would like to thank A.~Caraiani, D.~Jafferis, B.~Safdi, T.~Tesileanu, M.~Yamazaki and especially I.~Klebanov for discussion.
CH, DG, and SP were supported in party by the NSF under Grant No.~PHY-0756966.
CH and DG were also supported in part by the US NSF under Grant No.~PHY-0844827, and SP by Princeton University through a Porter Ogden Jacobus Fellowship.  CH also thanks the Sloan Foundation for partial support. 

\appendix

\section{Proof of the tree formula}
\label{app:treeproof}

We prove the Tree Formula by assuming that the columns $\beta_a$ are ordered in the way described in Corollary \ref{cor:n2ordered}.  Thus, the proof reduces to showing that the following equation holds
\be
\sum_{a=1}^d \frac{\gamma_{a(a+1)}}{\sigma_a \sigma_{a+1}} = 
2 \frac{\sum_{(V,E)\in T} \prod_{(a,b)\in E} \gamma_{ab}}{\prod_{a=1}^d \sigma_a} \ .
\ee

We have the ``Ptolemy relations''
\begin{equation} \label{eq:ptolemy}
\g_{ab} \g_{cd} = \g_{ac} \g_{bd} + \g_{ad} \g_{bc} \qquad \mathrm{if} \qquad ab//cd.
\end{equation}
The notation $//$ means $a$ and $b$ separate $c$ and $d$.
We can use the relations to show
\begin{equation} \label{eq:sums}
\g_{(a-1)(a+1)}\s_a + 2 \g_{(a-1)a} \g_{a(a+1)} = \g_{(a-1)a} \s_{a+1} + \g_{a(a+1)} \s_{a-1}.
\end{equation}

Our starting point is the Kirchhoff matrix-tree theorem
which gives a relation between the absolute value of a certain determinant and the sum over trees.  In particular, consider the matrix $Q^*$ where
\begin{equation}
Q^* = \left(
\begin{array}{ccccc}
\s_2 & -\g_{23} & -\g_{24} & \dots & -\g_{2d} \\
-\g_{32} & \s_3 & -\g_{34} & \dots & -\g_{3d} \\
-\g_{42} & -\g_{43} & \s_4 & \dots & -\g_{4d} \\
\vdots & \vdots & \vdots & dots & \vdots \\
-\g_{d2} & -\g_{d3} & -\g_{d4} & \dots & \s_d
\end{array}
\right).
\end{equation}
The matrix-tree theorem states that
\be
 \det Q^*  = \sum_{(V,E)\in T} \prod_{(a,b)\in E} \gamma_{ab} \ .
\ee
We observe that if we take $\g_{(a+1)a}$ times the $(a-2)$nd row
minus $\g_{(a-1)(a+1)}$ times the $(a-1)$st row plus
$\g_{a(a-1)}$ times the $a$th row, then most of the entries will cancel out.
So we have $AQ^*=B$ where
\begin{eqnarray}
A & = &
\left( \begin{array}{ccccc}
\g_{13} + \g_{23} & -\g_{12} + \g_{23} & \g_{23} & \dots & \g_{23} \\
-\g_{34} & \g_{24} & -\g_{23} & \dots & 0 \\
0 & -\g_{45} & \g_{35} & \dots & 0 \\
\vdots & \vdots & \vdots & dots & \vdots \\
\g_{(d-1)d} & \g_{(d-1)d} & \g_{(d-1)d} & \dots & \g_{(d-1)1} + \g_{(d-1)d}
\end{array} \right) \\
B & = &
\left( \begin{array}{ccccc}
\g_{12} \s_3 + \g_{23} \s_1 & -\g_{12} \s_3 & 0 & \dots & 0 \\
-\g_{34} \s_2 & \g_{23} \s_4 + \g_{34} \s_2 & -\g_{23} \s_4 & \dots & 0 \\
0 & -\g_{45} \s_3 & \g_{34} \s_5 + \g_{45} \s_3 & \dots & 0 \\
\vdots & \vdots & \vdots & dots & \vdots \\
0 & 0 & 0 & \dots & \g_{(d-1)d} \s_{1} + \g_{d1} \s_{d-1}
\end{array} \right)
\end{eqnarray}
In constructing $A$ we used the fact that the missing row of $Q$ is minus
the sum of all of the other rows.
We used \eqref{eq:ptolemy} and \eqref{eq:sums} to simplify $B$.
Tri-diagonal matrices satisfy a three term recurrence relation.
If we let $B_a$ be the matrix consisting of the first $a$ rows of $B$,
then
\begin{equation}
\det B_a = (\g_{a(a+1)} \s_{a+2} + \g_{(a+1)(a+2)} \s_a) \det B_{a-1} -
\g_{(a+1)(a+2)} \g_{(a-1)a} \s_a \s_{a+1} \det B_{a-2}.
\end{equation}
We can show by induction that
\begin{equation}
\det(B_a) = \left( \prod_{b=2}^{a} \g_{b(b+1)} \right) \left( \prod_{b=1}^{a+2} \s_b \right) \left( \sum_{b=1}^{a+1} \frac{\g_{b(b+1)}}{\s_b \s_{b+1}} \right)
\end{equation}
and in particular 
\begin{equation}
\det(B) = \det(B_{d-1}) = \frac{\s_1}{\g_{12} \g_{d1}} \left( \prod_{a=1}^d \s_a \g_{a(a+1)} \right) \left(\sum_{a=1}^d\frac{\g_{a(a+1)}}{\s_a \s_{a+1}}\right). 
\end{equation}
Therefore we want to show that $\det(A) = 2 \s_1 \prod_{a=2}^{d-1} \g_{a(a+1)}$.

We let $A'$ be $A$ with the $\g_{23}$ in the top row and the $\g_{(d-1)d}$ in
the bottom row removed.
Let $A_{ab}'$ be the matrix containing the $a$th through $b$th rows
and columns of $A'$.
Write the determinant of $A$ as
\[
\det(A) = (A'_1 + \g_{23} u) \wedge A'_2 \wedge \ldots \wedge A'_{d-2} \wedge (A'_{d-1} + \g_{(d-1)d} u)
\]
where $u = (1,1, \ldots, 1)$ and $A'_a$ are the rows of $A'$.  Expanding out the anti-symmetric product, we find that
\begin{eqnarray}
\nonumber
\det(A) & = & \det(A') 
+ \g_{23} u \wedge A'_2 \wedge \ldots \wedge A'_{d-1} + \g_{(d-1)d} A'_1 \wedge \ldots \wedge A'_{d-2} \wedge u
\\
& = & \det(A')+ \sum_{a=1}^d \left[  \g_{23}
\left(\prod_{b=1}^{a-1} \g_{(b+2)(b+3)} \right)
\det A_{(a+1)(d-1)}' \right. 
\nonumber \\ 
& & \left.
+ \g_{d(d-1)} \left(\prod_{b=a+2}^d \g_{(b-1)(b-2)}\right) \det A_{1(a-1)}' \right] \ .
\end{eqnarray}

Using the three term recurrence relation for the determinant of a tridiagonal matrix, we can show by induction that
 \es{DetAprime}{
\det A_{1a}' &=   \g_{1(a+2)} \prod_{b=2}^{a} \g_{b(b+1)} \ , \\
\det A_{a(d-1)}' &=  \g_{1a} \prod_{b=a}^{d-2} \g_{(b+1)(b+2)} \ , \\
\det A' &= 0 \,,
 }
where last equation follows from the second one by setting $a=1$ and using $\g_{11} = 0$.  Therefore
\es{GotdetA}{
\det(A) &= 2 \left(\sum_a \g_{1a}\right) \left( \prod_{a=2}^{d-1} \g_{a(a+1)} \right) \\
&= 2 \s_1 \prod_{a=2}^{d-1} \g_{a(a+1)} 
}
as expected.

\section{A virial theorem for matrix models}
\label{VIRIAL}

To leading order in $N$, the matrix model for superconformal field theories with M-theory duals involves extremizing a free energy functional of the type
 \es{FreeFunctional}{
  F[\rho, y_a] = \int dx\, \rho(x)^2 f(y_a(x)) - \int dx\, \rho(x) V(x, y_a(x)) \,,
 }
for some functions $f$ and $V$.   In all examples, $V$ is homogeneous of degree $1$ in $x$, namely it satisfies
 \es{VHom}{
  V = x \frac{\partial V}{\partial x} \,.
 }
The free energy functional \eqref{FreeFunctional} should be extremized under the constraint that $\rho$ is a density normalized so that $\int dx\, \rho(x) = 1$.  Such a constraint can be implemented with a Lagrange multiplier $\mu$ by defining a new functional
 \es{Ftilde}{
  \tilde F[\rho, y_a] = F[\rho, y_a] - 2 \pi N^{3/2} \mu \left(\int dx\, \rho(x) - 1 \right) 
 }
and varying \eqref{Ftilde} instead of \eqref{FreeFunctional} with respect to $\rho$ and $y_a$.  If we denote by $F$ the on-shell value of $F[\rho, y_a]$ (or of $\tilde F[\rho, y_a]$), we will now show that
 \es{FmuRelation}{
  F = \frac{4 \pi N^{3/2}}{3} \mu \,,
 }
regardless of which particular matrix model we're solving provided that eq.~\eqref{VHom} is obeyed.

\subsection{A slick proof}

The equations of motion following from \eqref{Ftilde} are
 \es{eoms}{
  2 \rho(x) f(y_a(x)) - V(x, y_a(x)) &= 2 \pi N^{3/2} \mu \,, \\
  \rho(x)^2 \partial_a f(y_a(x)) - \rho(x) \partial_a V(x, y_a(x))  &=0 \,,
 }
Differentiating the first equation with respect to $x$, multiplying it by $\rho(x)$, and using the second equation we get
 \es{eomsDer}{
  \rho(x) \bigg[ 2 \rho'(x) f(y_a(x))  + \rho(x)^2 \sum_a \partial_a f(y_a(x)) y_a'(x) \bigg] 
    = \rho(x) \frac{\partial}{\partial x} V(x, y_a(x))\,.
 }

To prove \eqref{FmuRelation}, consider the function
 \es{GDef}{
  G(x) = x \rho(x)^2 f(y_a(x)) \,.
 }
Let's calculate the derivative of this function with respect to $x$:
 \es{GDer}{
  G'(x) &= \rho(x)^2 f(y_a(x)) + 2 x \rho(x) \rho'(x) f(y_a(x))  + 
  x \rho(x)^2 \sum_a \partial_a f(y_a(x)) y_a'(x) \\
  &=\rho(x)^2 f(y_a(x)) + x \rho(x) \frac{\partial}{\partial x} V(x, y_a(x)) \\
  &=\rho(x)^2 f(y_a(x)) + \rho(x) V(x, y_a(x)) \,,
 }
where in the second line we used \eqref{eomsDer} and in the third line we used \eqref{VHom}.  The function $G(x)$ vanishes at $\pm \infty$, so $\int dx\, G'(x) = 0$, which from \eqref{GDer} implies that on-shell we have
 \es{Virial}{
  \int dx\, \rho(x)^2 f(y_a(x)) = -\int dx\, \rho(x) V(x, y_a(x)) \,.
 }
From \eqref{FreeFunctional} we have then that
 \es{Fonshell}{
  F = 2 \int dx\, \rho(x)^2 f(y_a(x)) \,.
 }
Multiplying the first equation in \eqref{eoms} by $\rho(x)$ and integrating in $x$ we get
 \es{muonshell}{
  2 \pi N^{3/2} \mu = 3 \int dx\, \rho(x)^2 f(y_a(x)) \,.
 }
Taking the ratio of the last two equations one obtains \eqref{FmuRelation}.

\subsection{A more enlightening proof}

The proof given above is really that of a virial theorem.  To put the virial theorem in a more familiar form, let's define the cumulative distribution $t(x) = \int_{-\infty}^x dx' \rho(x') \in [0, 1]$ and express \eqref{FreeFunctional} as the action
 \es{Action}{
  S[x(t), y_a(t)] =  \int dt\, L \,, \qquad L \equiv \frac{f(y_a)}{\dot{x}} -  V(x, y_a) \,,
 }
where $x$ and $y_a$ should be thought of as functions of $t$.  To go to a Hamiltonian formulation, we introduce the momentum conjugate to $x$:
 \es{pxDef}{
  p_x \equiv \frac{\partial L}{\partial \dot{x}} = -\frac{f(y_a)}{\dot{x}^2}  \,.
 }
The Hamiltonian is
 \es{Hamilt}{
  H \equiv p_x \dot{x} - L = \text{KE} + \text{PE} \,, \qquad
   \text{KE} \equiv  -2 \frac{f(y_a)}{\dot{x}} = -2 f(y_a)^{1/2} \sqrt{-p_x} \,, \qquad
  \text{PE} \equiv V(x, y_a) \,.
 }
In deriving a virial theorem for this Hamiltonian, one can just ignore the dependence on $y_a$ because the $y_a$ are non-dynamical.  Since the kinetic energy $\text{KE}$ is homogeneous of degree $1/2$ in $p_x$ and the potential energy is homogeneous of degree $1$ in $x$, we have
 \es{VirialHamilt}{
  \frac 12 \langle \text{KE} \rangle = \langle \text{PE} \rangle \,,
 }
where $\langle \cdots \rangle$ means $\int_0^1 dt\, ( \cdots )$.  From \eqref{Action} one sees that
 \es{onshellaction}{
  F \equiv S_{\text{on-shell}} = -\frac 12 \langle \text{KE} \rangle - \langle \text{PE} \rangle 
   = -\langle \text{KE} \rangle \,.
 }
From the first equation in \eqref{eoms} we see that $H = -2 \pi N^{3/2} \mu$, so
 \es{muonshellAgain}{
  \frac{N^{3/2}}{2 \pi} \mu = -\langle H \rangle = - \langle \text{KE} \rangle - \langle \text{PE} \rangle
  = -\frac 32 \langle \text{KE} \rangle \,.
 }
Taking the ratio of \eqref{onshellaction} and \eqref{muonshellAgain} one again obtains \eqref{FmuRelation}.

\bibliographystyle{ssg}
\bibliography{necklace}

\end{document}

%% file: IIBBranes.pdf_tex

\begingroup
  \makeatletter
  \providecommand\color[2][]{%
    \errmessage{(Inkscape) Color is used for the text in Inkscape, but the package 'color.sty' is not loaded}
    \renewcommand\color[2][]{}%
  }
  \providecommand\transparent[1]{%
    \errmessage{(Inkscape) Transparency is used (non-zero) for the text in Inkscape, but the package 'transparent.sty' is not loaded}
    \renewcommand\transparent[1]{}%
  }
  \providecommand\rotatebox[2]{#2}
  \ifx\svgwidth\undefined
    \setlength{\unitlength}{994.046875pt}
  \else
    \setlength{\unitlength}{\svgwidth}
  \fi
  \global\let\svgwidth\undefined
  \makeatother
  \begin{picture}(1,0.48435928)%
    \put(0,0){\includegraphics[width=\unitlength]{IIBBranes.pdf}}%
    \put(0.57217857,0.07360149){\color[rgb]{0,0,0}\makebox(0,0)[lb]{\smash{$x_6$}}}%
    \put(-0.00267216,0.38632028){\color[rgb]{0,0,0}\makebox(0,0)[lb]{\smash{$(p_1, q_1)$ 5-brane}}}%
    \put(0.2111723,0.45300297){\color[rgb]{0,0,0}\makebox(0,0)[lb]{\smash{$(p_2, q_2)$}}}%
    \put(0.40432217,0.45760177){\color[rgb]{0,0,0}\makebox(0,0)[lb]{\smash{$(p_3, q_3)$}}}%
    \put(0.54458572,0.45070357){\color[rgb]{0,0,0}\makebox(0,0)[lb]{\smash{$(p_4, q_4)$}}}%
    \put(0.8021189,0.45990118){\color[rgb]{0,0,0}\makebox(0,0)[lb]{\smash{$(p_d, q_d)$}}}%
    \put(0.09657493,0.00689834){\color[rgb]{0,0,0}\makebox(0,0)[lb]{\smash{$N$ D3's}}}%
  \end{picture}%
\endgroup

%% file: NecklaceQuiverFund.pdf_tex

\begingroup
  \makeatletter
  \providecommand\color[2][]{%
    \errmessage{(Inkscape) Color is used for the text in Inkscape, but the package 'color.sty' is not loaded}
    \renewcommand\color[2][]{}%
  }
  \providecommand\transparent[1]{%
    \errmessage{(Inkscape) Transparency is used (non-zero) for the text in Inkscape, but the package 'transparent.sty' is not loaded}
    \renewcommand\transparent[1]{}%
  }
  \providecommand\rotatebox[2]{#2}
  \ifx\svgwidth\undefined
    \setlength{\unitlength}{459.09067383pt}
  \else
    \setlength{\unitlength}{\svgwidth}
  \fi
  \global\let\svgwidth\undefined
  \makeatother
  \begin{picture}(1,0.96371043)%
    \put(0,0){\includegraphics[width=\unitlength]{NecklaceQuiverFund.pdf}}%
    \put(0.48965262,0.74725023){\color[rgb]{0,0,0}\makebox(0,0)[lt]{\begin{minipage}{0.18172568\unitlength}\raggedright $k_d$\end{minipage}}}%
    \put(0.69981928,0.53588539){\color[rgb]{0,0,0}\makebox(0,0)[lt]{\begin{minipage}{0.10206513\unitlength}\raggedright $k_1$\end{minipage}}}%
    \put(0.55866307,0.27463203){\color[rgb]{0,0,0}\makebox(0,0)[lt]{\begin{minipage}{0.12944842\unitlength}\raggedright $k_2$\end{minipage}}}%
    \put(0.23734499,0.30091838){\color[rgb]{0,0,0}\makebox(0,0)[lt]{\begin{minipage}{0.13442727\unitlength}\raggedright $k_3$\end{minipage}}}%
    \put(0.18327759,0.61644832){\color[rgb]{0,0,0}\makebox(0,0)[lt]{\begin{minipage}{0.1493636\unitlength}\raggedright $k_{d-1}$\\ \end{minipage}}}%
    \put(0.22391427,0.77471426){\color[rgb]{0,0,0}\makebox(0,0)[lb]{\smash{$A_d$}}}%
    \put(0.66453679,0.75230977){\color[rgb]{0,0,0}\makebox(0,0)[lb]{\smash{$A_1$}}}%
    \put(0.76411254,0.28679326){\color[rgb]{0,0,0}\makebox(0,0)[lb]{\smash{$A_2$}}}%
    \put(0.36580965,0.09262058){\color[rgb]{0,0,0}\makebox(0,0)[lb]{\smash{$A_3$}}}%
    \put(0.54255659,0.55813708){\color[rgb]{0,0,0}\makebox(0,0)[lb]{\smash{$B_1$}}}%
    \put(0.56993987,0.39881589){\color[rgb]{0,0,0}\makebox(0,0)[lb]{\smash{$B_2$}}}%
    \put(0.40563995,0.3091978){\color[rgb]{0,0,0}\makebox(0,0)[lb]{\smash{$B_3$}}}%
    \put(0.3434051,0.58552037){\color[rgb]{0,0,0}\makebox(0,0)[lb]{\smash{$B_d$}}}%
    \put(0.12938272,0.13697567){\color[rgb]{0,0,0}\makebox(0,0)[lt]{\begin{minipage}{0.13442727\unitlength}\raggedright $n_3$\end{minipage}}}%
    \put(0.01142396,0.73676607){\color[rgb]{0,0,0}\makebox(0,0)[lt]{\begin{minipage}{0.13442727\unitlength}\raggedright $n_{d-1}$\end{minipage}}}%
    \put(0.48925695,0.94869199){\color[rgb]{0,0,0}\makebox(0,0)[lt]{\begin{minipage}{0.13442727\unitlength}\raggedright $n_d$\end{minipage}}}%
    \put(0.90711091,0.53083801){\color[rgb]{0,0,0}\makebox(0,0)[lt]{\begin{minipage}{0.13442727\unitlength}\raggedright $n_1$\end{minipage}}}%
    \put(0.65719825,0.06300152){\color[rgb]{0,0,0}\makebox(0,0)[lt]{\begin{minipage}{0.13442727\unitlength}\raggedright $n_2$\end{minipage}}}%
  \end{picture}%
\endgroup